%% file: main.tex
\documentclass[12pt]{article}
\usepackage[margin=1in]{geometry}
\usepackage[utf8]{inputenc}
\usepackage[T1]{fontenc}
\usepackage{times}

\usepackage{diagbox}
\usepackage{soul}
\usepackage{url}
\usepackage[hidelinks]{hyperref}
\usepackage{amsmath,amsthm,amssymb}
\usepackage{thmtools,thm-restate}
\usepackage{booktabs}
\usepackage{helvet}  
\usepackage{courier}  
\usepackage{graphicx}  
\usepackage{nicefrac}

\usepackage{amsfonts,pifont}
\usepackage{footnote}
\usepackage[numbers]{natbib}
\usepackage{cleveref}
\usepackage{mathtools}

\usepackage{algorithm}
\usepackage[noend]{algpseudocode}
\usepackage{arydshln}
\usepackage{boxedminipage}
\usepackage{calrsfs}
\usepackage{comment}
\usepackage{enumitem,lipsum}
\usepackage{tikz}
\usepackage{adjustbox}
\usetikzlibrary{positioning,shapes,arrows,graphs}
\usepackage[flushleft]{threeparttable}
\usepackage[textwidth=41mm,textsize=small]{todonotes}
\usepackage{wrapfig}
\usepackage{xspace}
\usepackage{pifont}     
\usepackage[normalem]{ulem}
\usepackage{authblk}
\usepackage{pgfplots}

\newtheorem{definition}{Definition}
\newtheorem{lemma}{Lemma}

\newtheorem{theorem}{Theorem}
\newtheorem{example}{Example}
\newtheorem{remark}{Remark}

\newenvironment{FitToWidth}[1][\textwidth]{
\begin{adjustbox}{width=#1,center}
}
{\end{adjustbox}}

\newcommand{\rurl}[1]{
  \href{https://#1}{\nolinkurl{#1}}
}

\definecolor{MyGreen}{rgb}{0, 0.7, 0}
\definecolor{MyRed}{rgb}{0.8, 0, 0}

\newcommand{\myOmit}[1]{}

\newcommand{\SW}{\operatorname{SW}}
\newcommand{\maxpe}{\operatorname{MaxPE}}
\newcommand{\pe}{\operatorname{PE}}

\usepackage{color}

\newtheorem{prop}{Proposition}

\begin{document}
\title{Coordinating Monetary Contributions in \\Participatory Budgeting}

\author[1]{Haris Aziz\thanks{haris.aziz@unsw.edu.au}}
\author[2]{Sujit Gujar\thanks{sujit.gujar@iiit.ac.in}}
\author[2]{Manisha Padala\thanks{manisha.padala@research.iiit.ac.in}}
\author[1]{\\Mashbat Suzuki\thanks{mashbat.suzuki@unsw.edu.au}}
\author[1]{Jeremy Vollen\thanks{j.vollen@unsw.edu.au}}
\affil[1]{UNSW Sydney}
\affil[2]{IIIT Hyderabad}
\date{}

\maketitle
\begin{abstract}
 We formalize a framework for coordinating funding and selecting projects, the costs of which are shared among agents with quasi-linear utility functions and individual budgets. Our model contains the classical discrete participatory budgeting model as a special case, while capturing other useful scenarios. We propose several important axioms and objectives and study how well they can be simultaneously satisfied. We show that whereas welfare maximization admits an FPTAS, welfare maximization subject to a natural and very weak participation requirement leads to a strong inapproximability. This result is bypassed if we consider some natural restricted valuations, namely laminar single-minded valuations and symmetric valuations. Our analysis for the former restriction leads to the discovery of a new class of tractable instances for the Set Union Knapsack problem, a classical problem in combinatorial optimization.
\end{abstract}

\section{Introduction}
\label{sec:intro}

\emph{Participatory budgeting} (PB) is an exciting grassroots democratic paradigm in which members of a community collectively decide on which relevant public projects should be funded~(see, e.g., \citep{AzSh20a,shah2007,VNS20a}).
The funding decisions take into account the preferences and valuations of the members. 

{The positive influence of PB is apparent in the many implementations across the globe at the country, city and community level. For example, a PB scheme in the Govanhill area of Glasgow, Scotland empowered local residents to direct funds towards imperative projects like addiction family support groups, a community justice partnership, and refurbishment of locally significant public baths.\footnote{\rurl{local.gov.uk/case-studies/govanhill-glasgow}} In the 2014-15 New York City PB process, 51,000 people voted to fund \$32 million of neighbourhood improvements.\footnote{\rurl{hudexchange.info/programs/participatory-budgeting}} As both examples demonstrate, efficient use of public funds and a significant improvement in community member involvement are two of the foremost advantages of PB processes around the world. 
}

{In all implementations of PB that we are aware of, the process relies upon a central authority to determine and provide a budget to fund projects. In practice, this requirement excludes groups who lack the institutional structure required to pool resources from initiating a PB process. For example, several neighbouring municipalities may wish to collaborate on funding projects which can benefit residents from multiple communities simultaneously. On a smaller scale, a group of flatmates may need to decide which furnishings and appliances to buy. Or, a number of organisations co-hosting an event may need to select a list of speakers, each of which charges their own fee. In each of these cases, while PB seems a natural process of arriving at a mutually beneficial outcome, the classical discrete PB model is insufficient because PB's focus is on project selection and not efficient resource pooling. In this paper, we present a framework which captures both of these components simultaneously.}

{Furthermore, in the classical PB model, it is typical to assume that agents' utilities depend only on their valuations for selected projects and that agents are indifferent toward the amount of the budget used. While this is appropriate in that setting because the agents do not necessarily believe leftover budget will benefit them, this is not the case in our setting since agents can use their leftover funds directly. For this reason, we model agents with utilities dependent upon their contributions (i.e. \textit{quasi-linear utilities}).
}

In summary, we consider a flexible and general framework which we term \textit{PB with Resource Pooling} in which (i) agents can have their own budgets, which can be directed to fund desirable projects, and (ii) agents care both about which projects are funded, and how much monetary contribution they make due to quasi-linear utilities. We point out that our framework captures the classical discrete PB model in the following way: an artificial agent can be introduced with budget equal to the central budget and valuations for project bundles equal to their total costs.

Taking cues from what we see as the key successes of PB implementations, we primarily focus on the \emph{efficiency} of the outcome while ensuring each involved agent benefits from the outcome and is thus incentivized to participate in the process. Towards this, we study the possibility of achieving optimal \emph{utilitarian welfare} alongside a very basic participation notion, namely \emph{weak participation}, that guarantees positive utility to all of the agents involved. 

\paragraph{\textbf{Contributions}}
Our first contribution is a meaningful model of collective decision making that connects various problems including participatory budgeting, cost sharing problems, and crowd-sourcing. The model can also be viewed as a bridge between voting problems and mechanism design with money.

{In \Cref{sec:prelim}, we lay the groundwork for axiomatic research for the problem by formalizing meaningful axioms that capture efficiency and participation incentives. To capture a PB program which is both useful and sustainable, we focus on mechanisms that satisfy \textit{weak participation}, which requires that every agent benefits from participation.}

In \Cref{sec::adp}, we show that whereas welfare maximization admits a \textit{fully} \textit{polynomial time approximation scheme} (FPTAS)\footnote{An algorithm which approximates the optimal solution by a factor of at least $1-\epsilon$ in time polynomial in the instance size and $1/\epsilon$ for any $\epsilon>0$.}, welfare maximization subject to weak participation leads to a strong inapproximability for as few as two agents. 
We then show that the same objective is inapproximable even for the case with identical project costs despite there being an exact, polytime algorithm for welfare maximization. 

In \Cref{sec:laminar_sm}, we give an FPTAS for welfare maximization subject to participation in the single-minded setting with laminar demand sets, utilizing a result from the \textit{knapsack with conflict graphs} problem. 
We use the insight from this argument to reveal a tractable case for the \textit{set-union knapsack problem}, making a novel contribution to a well-studied generalization of the knapsack problem (\Cref{app:sukp}). 
We also present a polytime algorithm for \emph{symmetric valuations} (\Cref{sec:symm}) and experimentally demonstrate that a greedy approach achieves close to optimal welfare on an average basis, even when subjected to our participation requirement (\Cref{sec:experiments}). 
Due to space constraints, we defer some proofs to the appendix. Our results are summarized in \Cref{tab:summary}. 

\begin{table*}
    \centering
    \caption{Summary of our computational results}  
    \scalebox{0.9}{
    \begin{threeparttable}
    \begin{tabular}{c|l|l}
        \toprule
         & Welfare Maximization &  Welfare Maximization subject to   \\   & & \ \ \ \ \ \ \ \ \ \ Weak Participation \\
        \midrule
        General case & NP-hard; FPTAS & Inapproximable* for any $n\geq 2$\\
        Identical costs & Polynomial-time exact & Inapproximable* \\
        Laminar Single-minded valuations & NP-hard & FPTAS \\
        Symmetric valuations & Polynomial-time exact & Polynomial-time exact \\
        \bottomrule
    \end{tabular}
    \begin{tablenotes}
      \item[*] Inapproximability results hold assuming $P\neq NP$.
    \end{tablenotes}
    \end{threeparttable}}
    \label{tab:summary}
\end{table*}

\section{Related Work}
\noindent\textbf{Funding Public Goods/Projects.}
Many classical papers from public economics focus on mechanisms for public good provision \citep{samuelson1954pure,Grov73a}. 
However, they are primarily concerned with incentives and mostly focus on divisible public goods \citep{GrLa79a}, a setting which differs from ours significantly. 
\citet{BHW19a} presented a mechanism with quasi-linear utilities, but they consider divisible goods and do not include individual budgets. 
\citet{BBGPSS22} considered a simple model without costs and without quasi-linear utilities. 
The model is more suitable for making donations to long-term projects. 
Similarly, \citet{wagnervcg} considered projects without fixed value and cost. 
They use the VCG mechanism to elicit agents' preferences truthfully and consider a different utility model from ours. 
\citet{AzGa21a} considered a model that allows for personal contributions and costs. 
However, their perspective is that of charitable coordination without the quasi-linear utility assumption. 

\noindent\textbf{Discrete Participatory Budgeting.} 
The model we formalize has some connections with discrete PB models~\citep{AzSh20a,ALT18a,NtPf19,PaTa20} in which agents do not make contributions. 
\citet{BNPS17a} and \citet{lu2011budgeted} study various aspects of preference elicitation in PB. 
These models generalize multi-winner voting~\citep{EFSS17a}. 
The special case of our model with identical project costs is similar to multi-winner elections with a variable number of winners ~\citep{Kilg16a}, with increased generality due to individual budgets and quasi-linear utilities. 
Some research has considered PB settings in which agents express budget constraints in addition to preferences \citep{Hersh2021,chen2022participatory}. 
Unlike in these works, the agents in our model care how much of their budget is used, better capturing the budgeting trade-offs considered by agents with limited resources.

\noindent\textbf{Cost-sharing/Crowdfunding.} 
PB without budget constrained agents overlaps significantly with the study of cost-sharing mechanisms \citep{moulin1994serial,tim2018}, an area focused on designing truthful and efficient mechanisms for sharing the cost of availing a certain service or project among the members. 
Specifically, PB can be likened to a cost-sharing mechanism  with non-rivalrous projects, i.e. each project has a fixed  cost \citep{ohseto2000,guo}. 

Some work in this literature considers truthful mechanisms for combinatorial cost sharing across multiple projects \citep{birmpas2019cost,shahar17}. 
\citet{birmpas2019cost} focused on symmetric submodular valuations and provided a mechanism which approximates social cost. 

The social cost considered is equivalent to the quasi-linear utility that we consider in our setting, but these works focus on agents without budget restrictions.
Additionally, utility-based PB bears resemblance to certain civic crowdfunding models \citep{damle2019civic,yan2021optimal,zubrickas2014provision}. 
However, these papers consider neither budget constraints nor the multiple project case. In this line of work, the goal is to analyze what happens at equilibrium when the agents are strategic, as opposed to our own goal of determining a welfare maximizing subset of projects to fund provided budget constraints.

\section{Preliminaries}
\label{sec:prelim}
 
In this work, we provide a framework for collective funding decisions in which agents derive quasi-linear utility and lack an exogenously defined shared budget. 
The setting has the following components.

\subsection{Model}

An instance of \textit{PB with Resource Pooling} contains a set of participating agents $N=[n]$ and a set of projects, $M=[m]$.
$\mathbf{C} = (C_j)_{j\in M}$ denote the project costs. Abusing notation slightly, we let $C(S)=\sum_{j\in S} C_j$ for any $S\subseteq M$. Each agent $i$ has valuation $v_i: 2^M \rightarrow \mathbb{R}_{\ge 0}$ and  a budget $b_i$, which is the maximum amount they are willing to contribute. We denote $\mathbf{v} = (v_i)_{i\in N}$ and $\mathbf{b} = (b_i)_{i\in N}$. We assume $v_i$ to be monotonic, that is $v_i(S) \leq v_i(T), S\subseteq T$. Agent $i$'s valuation of a single project $j$ is denoted by $v_{ij}$ for ease. We assume that the overall valuation of each project is higher than its cost, i.e. $\sum_{i\in N} v_{ij} \ge C_j, \ \forall j\in M$, since it will be trivially excluded otherwise.

 In summary, a PB with Resource Pooling instance is a tuple $I = \langle N, M, \mathbf{v}, \mathbf{b}, \mathbf{C} \rangle$. 
 
 A mechanism takes an instance $I$ and computes an outcome $(W, \mathbf{x})$, comprising a set of funded projects $W\subseteq M$ and a vector of agent payments $\mathbf{x}=(x_i)_{i\in N}$. Each agent then derives \emph{quasi-linear utility}, i.e. $u_i(W,x) = v_i(W) - x_i$. We close with an illustrative motivating example, which assumes additive valuations.

\begin{example}\label{ex:towns}

  Three towns perform cost-benefit analyses for three proposed projects: a concert hall, a shelter, and a pool. The costs, budgets, and valuations are given in \Cref{tab: example}.
\begin{table}[h!] 
  \centering
  \scalebox{0.9}{
  \begin{tabular}{cc|c|c|c}
  \multicolumn{1}{l}{}            & \multicolumn{1}{l|}{} &
  \multicolumn{1}{l|}{Hall} & \multicolumn{1}{l|}{Shelter} & \multicolumn{1}{l|}{Pool} \\ \cline{2-5} 
  \multicolumn{1}{c|}{}           & \diagbox{Budget}{Cost}           & 5                           & 4    & \multicolumn{1}{c|}{2}   \\ \hline
  \multicolumn{1}{c|}{Town A} & 2               & 2                         & 1                           & \multicolumn{1}{c|}{2} \\ \hline
  \multicolumn{1}{c|}{Town B} & 3               & 1                         & 2                           & \multicolumn{1}{c|}{2} \\ \hline
  \multicolumn{1}{c|}{Town C} & 1               & 4                         & 3                           & \multicolumn{1}{c|}{1} \\ \hline
  \end{tabular}}
  \caption{\label{tab: example} An example of \emph{PB with Resource Pooling} with three agents and three projects. The table records the costs, budgets, and valuations.}
  \vspace{-1em}
  \end{table}

  While Towns A and B have the funds to pay for the pool on their own, they will receive zero utility from doing so. Furthermore, no town can fund either of the other two projects on its own. Together, the three towns have just enough funds to build the shelter and the pool, the result of which will be a strictly positive utility for all towns. Lastly, we point out that, while the towns could instead fund the concert hall and generate positive social welfare, Town B will be required to pay at least 2, which exceeds its valuation and it will thus receive negative utility.
\end{example}

\subsection{Desirable {Axioms}}
\label{sec:dp}

We now formalize basic feasibility criteria and desirable properties related to participation incentives and efficiency.

\noindent\textbf{Feasibility.} We say that an outcome is \emph{feasible} if (1) no agent's payment exceeds their individual budget, i.e. $x_i\leq b_i$ for all $i\in N$, and (2) it is weakly budget balanced, a property which requires an outcome does not run a deficit and which we will now define.

\begin{definition}[Weakly Budget Balanced (WBB)] 
    An outcome $(W,\mathbf{x})$ is WBB if $C(W) \leq \sum_{i\in N} x_{i}$. When $C(W) = \sum_{i\in N} x_{i}$, we say the outcome is Budget Balanced (BB). 
\end{definition}

Henceforth, we focus our attention on outcomes which satisfy these basic requirements and will thus refer to feasible outcomes merely as outcomes.

\noindent\textbf{Participation.} We now turn our attention to properties that incentivize agents to participate in the mechanism. Weak Participation (WP) ensures that each agent obtains non-negative utility from the mechanism.
\begin{definition}[Weak Participation (WP)]
    An outcome $(W,\mathbf{x})$ is WP if $\forall i \in N$, $u_i(W, \mathbf{x}) \ge 0$, i.e., $v_i(W) \ge x_i$. 
\end{definition}

We note that WP is a minimal requirement and considerably weaker than 
Individual Rationality (IR), which ensures that the utility each agent obtains from the mechanism is at least the maximum utility they can obtain with only their budget. 
In other words, whereas IR captures individual utility maximization, WP ensures that no agent is making a loss. 

However, the gap between the social welfare of the best WP outcome and that of the best IR outcome can be arbitrarily large (refer to \Cref{ex:sacrifice}). Furthermore, checking whether an outcome is WP is in P, whereas checking whether an outcome is IR is NP-hard, even for a single agent. Hence, we will focus on welfare maximization within the space of WP outcomes. The following lemma illustrates that WP and feasibility can be captured simultaneously by a single inequality.

\begin{lemma}
\label{lem: WP-WBB}
  A set of projects $W\subseteq M$ can be funded in a feasible and WP manner if and only if 
  $$
  C(W) \leq \sum_{i\in N} \min(b_i,v_i(W) )
  $$ 
\end{lemma}
\begin{proof}
  Let $W\subseteq M$ be a set of projects which can be funded in a feasible and WP manner. That is, there exists payments $\mathbf{x}$ such that $x_i\leq b_i$ and $x_i \leq v_i(W)$ for each $i\in N$. Combining the two inequalities, $x_i\leq \min(b_i,v_i(W))$ for all $i\in N$. Since $(W,\mathbf{x})$ is WBB, we have that $C(W) \leq \sum_{i\in N}x_i$, and thus $C(W) \leq \sum_{i\in N} \min(b_i,v_i(W) )$, as desired.
    
  As for the other direction, we need only show that there exist payments $\mathbf{x}$ such that $(W,\mathbf{x})$ satisfies feasibility and WP. Consider payments $x_i=\min(b_i, v_i(W))$ for each $i\in N$. The payments are clearly feasible and WBB holds since $C(W) \leq \sum_{i\in N}\min(b_i, v_i(W)) = \sum_{i\in N} x_i$. WP is satisfied since $x_i = \min(b_i, v_i(W))\leq v_i(W)$ for all $i\in N$.
\end{proof}

\noindent\textbf{Efficiency.} We measure the efficiency of the outcome with respect to the utilitarian welfare, i.e. the sum of utilities. Since we model quasi-linear utility for each agent, the utilitarian social welfare (which we refer to simply as social welfare) is given by 
\begin{equation} \label{eq:usw1}
     \SW(W,x) = \sum_{i\in N} u_i(W,x) = \sum_{i\in N} v_i(W) - \sum_{i\in N} x_i 
\end{equation}

The utilitarian welfare optimal (UWO) outcome is the social welfare maximizing feasible outcome.
\begin{definition}[Utilitarian Welfare Optimal (UWO)]
  \label{def:uwo}
    An outcome $(W, \mathbf{x})$ is UWO if it maximizes social welfare
    as given by \Cref{eq:usw1}. 
\end{definition}
    
We define UWO-WP to represent the welfare-optimal outcome among outcomes that ensure WP.
\begin{definition}[Welfare Optimal among Weak Participation (UWO-WP)]
    \label{def:uwo-wp}
    
    An outcome $(W, \mathbf{x})$ is UWO-WP if it is welfare optimal among WP outcomes, i.e. $\forall i\in N, v_i(W) \ge x_i$.
\end{definition}

\begin{remark} \label{rmk:bb}
 For any outcome $(W,x)$ which is WBB, there is an outcome which is BB, maintains WP, and achieves weakly greater social welfare.
\end{remark}

To see why \Cref{rmk:bb} is true, consider for any WBB outcome an alternative outcome with the same project selection, but with some set of agents' payments decreased such that the outcome is BB. By \Cref{rmk:bb} and our objectives stated in Definitions \ref{def:uwo} and \ref{def:uwo-wp}, we can restrict our attention to outcomes which are BB ($\sum_{i\in N}x_i=C(W)$) and refer to social welfare from now on as the following:

\begin{equation} \label{eq:usw2}
     \SW(W) = \sum_{i\in N} v_i(W) - C(W)
\end{equation}

In the next section, we prove that the problem of computing a 
UWO-WP outcome is inapproximable, assuming $P\neq NP$.

\section{Inapproximability of UWO-WP}

\label{sec::adp}
{In this section, we first show the hardness of finding a welfare-optimal project bundle, i.e. UWO. We then show that the same objective with the additional constraint of WP, i.e. UWO-WP, is inapproximable.} Our inapproximability result is striking:
the problem of finding UWO admits an FPTAS but after imposing the very weak requirement of WP, the same problem does not admit any polynomial approximation guarantees, even in the setting with only two agents.

The results in this section focus entirely on the setting with \emph{additive valuations}, under which $v_i(S) = \sum_{j \in S} v_{ij}$ for any $S\subseteq M$. We point out that the negative results that follow also hold in settings with any class of valuations broader than additive (e.g. superadditive, subadditive, general setting). We first make some remarks on computation of UWO (\Cref{def:uwo}). Due to \Cref{eq:usw2}, the UWO problem is given by
\begin{align*}
\max_{W\subseteq M} \ \ \ \  &\sum_{i\in N} v_i(W) - C(W)\\
& s.t. \hspace{1em} C(W) \leq \sum_{i\in N}b_i.  
\end{align*}  

\begin{restatable}{remark}{rmkuwo}
\label{rmk:uwo_fptas}
 The UWO problem is NP-hard even for a single agent. 
 However, in the setting with additive valuations, UWO admits an FPTAS \citep{fptas}.
\end{restatable}

Despite the positive approximation results, approximation guarantees to UWO may be possible only with huge costs to certain agents who incur large disutility, thus motivating the search for outcomes which satisfy WP. However, it is impossible to guarantee an outcome that satisfies both UWO and WP\footnote{Consider an instance with two agents and one project with $C_1=1$, valuations $v_{11}= 2$ and $v_{21}=0$, and budgets $b_1=0$ and $b_2 =1$. The UWO outcome is to fund the project, but this violates WP.}. Thus, we are interested in finding solutions which are welfare optimal among those which satisfy WP, i.e. that are UWO-WP (\Cref{def:uwo-wp}). However, the proof of \Cref{rmk:uwo_fptas} (\Cref{app:uwo_fptas}) can be used to show that finding a UWO-WP allocation is NP-hard, even for one agent.

\subsection{Gap Introducing Reduction for UWO-WP}
\label{subsec:gap_intro}

{As finding UWO-WP is NP-hard, we look for an approximation algorithm.} To see why finding such an approximation algorithm will be difficult, consider the following example.

\begin{example}\label{ex:sacrifice}
\begin{table}[htb]
\centering
\scalebox{0.9}{
  \begin{tabular}{cc|c|c|c|c|}
                               &             & P1 & P2 & P3 & P4 \\ \cline{2-6} 
  \multicolumn{1}{c|}{}        & \diagbox{Budget}{Cost} & 1      & 2    & 1   & 1  \\ \hline
  \multicolumn{1}{c|}{Agent 1} & 2           & 0      & 20     & 2-$\epsilon$    & 2 \\ \hline
  \multicolumn{1}{c|}{Agent 2} & 0           & H      & 0      & 20     & 0 \\ \hline
  \end{tabular}}
\caption{An instance with $n=2$ and $m=4$ where the UWO-WP outcome must contain a unique set of projects (P4) to allow the WP funding of a high-welfare project (P1).}\label{tab:sacrifice}
\vspace{-1em}
\end{table}
Observe that in the example given in \Cref{tab:sacrifice}, agent 1 holds the entire collective budget and maximizes her utility by funding project 2, which can be funded in a WP manner. But, no budget remains to fund other projects if project 2 is selected. Project 3, on the other hand, gives positive utility to both agents and leaves enough budget to fund project 4 in a WP manner (but not project 1). 

However, note that $W=\{1,4\}$ can be funded in a WP manner and achieves $SW(W)=H$. Although project 4 seems inferior to projects 2 or 3 on its own, it must be selected in order to fund project 1. Since H can be arbitrarily high, any algorithm which hopes to give a good approximation of UWO-WP should be able to find the set of projects that makes project 1 affordable if such a project set exists. 
\end{example}

Inspired by \Cref{ex:sacrifice}, where excess money must be extracted optimally to fund valuable projects, we introduce a related problem.

\noindent\textsc{Maximum Excess Payment Extraction} ($\maxpe$):
\begin{align*}
\max_{W\subseteq M} \pe(W) =  \sum_{i\in N} \min(b_i,v_i(W) )- C(W)
\end{align*}  

The above quantity measures the amount of excess money we can extract after paying for the costs of the projects in a WP manner. The decision version of the $\maxpe$ problem asks if $\maxpe\geq t$, that is whether $t$ dollars can be extracted while maintaining WP. Note that the decision version of the problem can be used to compute $\maxpe$ simply by using binary search. The following theorem states that any polynomial approximation for UWO-WP can be used to solve the $\maxpe$ decision problem in polynomial time.

\begin{theorem}\label{thm:gap_reduction}
  Let $f(n,m)$ be a polynomial time computable function. Any polynomial time algorithm with $f(n,m)$-approximation guarantee for UWO-WP  can be used to decide $\maxpe$ in polytime.
\end{theorem}
  \begin{proof}
      Consider an arbitrary instance $I=\langle N, M, \mathbf{v}, \mathbf{b}, \mathbf{C} \rangle$ of the $\maxpe\geq t$ problem where there are $n$ agents and $m$ projects. Let $\operatorname{Opt}(I)$ be the social welfare obtained from the UWO-WP solution to instance $I$. We know that $\operatorname{Opt}(I)\leq \sum_{i\in N}v_i(M) $. We now create a new instance by adding an agent $0$ and a project $0$ to $I$ such that:
    \begin{itemize}
    \item $v_{00}=(2f(n+1,m+1)) \sum_{i\in N}v_i(M) $, budget $b_0=0$ 
    \item $v_{0j}=0$ for all $j\in M$
    \item $v_{i0}=0$ for all $i\in N$
    \item $C_0 = t  $
    \end{itemize}   
    We denote the transformed instance as $I^t$. It has $n+1$ agents and $m+1$ projects. Note that the transformation is polynomial time as $f(n,m)$ is polynomial-time computable. 
    The modified instance creates a large gap in the optimal social welfare depending on the cost of project $0$, as seen in Figure~\ref{fig:inapprox}. As the figure illustrates, any $f$-approximation algorithm $\operatorname{Alg}$ will maintain such a gap and can thus be used to retrieve information about $\maxpe$. As we will see, this observation can be used to show approximation hardness.
    
    \begin{figure}
    \hspace{-3em}
    \begin{FitToWidth}[0.8\columnwidth]
    \begin{tikzpicture}[blacklinenode/.style={shape=rectangle, inner sep=0pt, minimum height=0.3pt, minimum width=6pt, draw},
              shadednode/.style={shape=rectangle, fill=black!20}]
      \def\x{8}
      \def\y{5}
      
      \draw[thick,->] (0,0) -- (\x,0) node[anchor=north west] {$t$};
      \draw[thick,->] (0,0) -- (0,\y) node[anchor=south east] {$\SW$};

      \draw (2pt, \y/5) -- (-2pt, \y/5) node[anchor=east] {$\sum_{i\in N} v_i(M)$};
      \draw (2pt, 2*\y/5) -- (-2pt, 2*\y/5) node[anchor=east] {$2 \sum_{i\in N} v_i(M)$};
      \draw (0, 4*\y/5) -- (0, 4*\y/5) node[anchor=east] {$2 f(n+1,m+1) \sum_{i\in N} v_i(M)$};
      \draw (\x/2, 2pt) -- (\x/2, -2pt) node[anchor=north] {$\maxpe(I)$};

      \draw[dashed, thick] (\x/2, 0) -- (\x/2, \y) ;

      \draw[thick] (0, 9*\y/10) -- (\x/2, 9*\y/10) ;
      \draw[thick] (\x/2, \y/5 - 0.25) -- (\x, \y/5 - 0.25) ;

      \fill[black!20] (0.02, 2*\y/5) rectangle (\x/2-0.01, 9*\y/10 - 0.01);
      \fill[black!20] (\x/2+0.01, \y/15) rectangle (\x, \y/5 - 0.26);

      \matrix [draw,below left] at (current bounding box.north east) {
        \node [blacklinenode,label=right:$\operatorname{Opt}(I^t)$] {}; \\
        \node [shadednode,label=right:$\operatorname{Alg}(I^t)$] {}; \\
      };
    \end{tikzpicture}
    \end{FitToWidth}
    \caption{The modified instance $I^t$ creates a large gap in optimal social welfare, depending on $t$, as shown by the solid lines. This welfare gap will remain for any $f$-approximation algorithm $\operatorname{Alg}$, as shown by the shaded regions, and illustrates how such an approximation can be used to solve the $\maxpe$ problem.}
    \label{fig:inapprox}
    \end{figure}
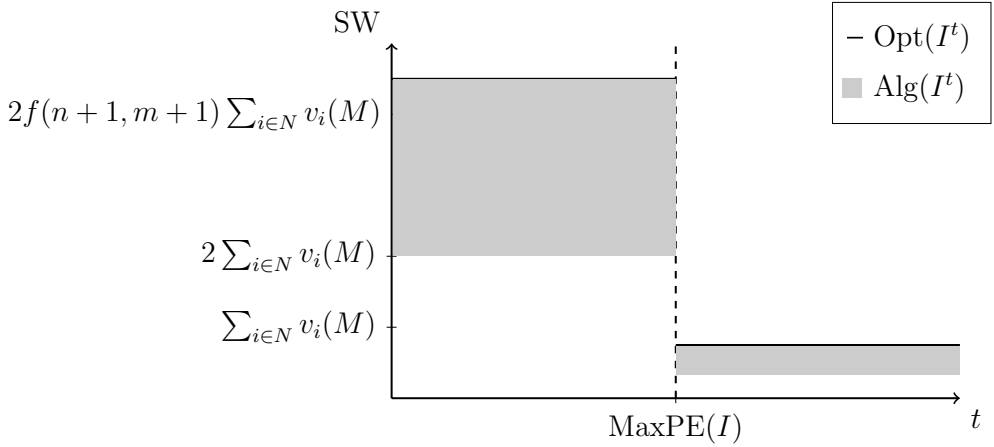

    Consider now a polynomial time algorithm $\operatorname{Alg}$ for UWO-WP with $f$-approximation guarantee, that is 
    \begin{equation}\label{eq: approx}
        f(n+1,m+1)\SW(\operatorname{Alg}(I^t))\geq \operatorname{Opt}(I^t).
    \end{equation}
    
    We now claim that $\maxpe(I)\geq t$ accepts (i.e., $t$ dollars can be extracted from agents in $I$ while satisfying WP) if and only if  $\SW(\operatorname{Alg}(I^t))\geq 2 \sum_{i\in N} v_i(M) $. The forward direction follows since if $t$ dollars can be extracted from agents in $I$ then project $0$ can be afforded. Hence, $\operatorname{Opt}(I^t)\geq (2f(n+1,m+1))\sum_{i\in N} v_i(M)$ since there is a WP solution which funds project zero. Thus, combining with \Cref{eq: approx}, we see that $\SW(\operatorname{Alg}(I^t))\geq 2 \sum_{i\in N} v_i(M)$.
    
    As for the backward direction, first observe that if $\SW(\operatorname{Alg}(I^t))\geq 2 \sum_{i\in N} v_i(M) $ then $0\in \operatorname{Alg}(I^t)$. To see this, suppose on the contrary $0\not \in \operatorname{Alg}(I^t)$. Then,  $\SW(\operatorname{Alg}(I^t))\leq \operatorname{Opt}(I) \leq \sum_{i\in N} v_i(M)$, leading to a contradiction. Thus, we have $0\in \operatorname{Alg}(I^t)$, which implies that the agents in $I$ must be able to pay for project $0$ in a WP manner. This is due to the fact that agent $0$ cannot contribute towards funding project $0$ since she has zero budget, and agents in $I$ do not value project $0$. Hence, $t$ dollars can be extracted from agents in $I$ in a WP manner. Thus, we may use $\operatorname{Alg}$ to decide $\maxpe\geq t$ in polynomial time. 
    \end{proof}

 In the subsections that follow, we use \Cref{thm:gap_reduction} to show that unlike UWO, which admits an FPTAS, UWO-WP does not admit \textit{any} bounded polynomial approximation guarantees in the additive setting, even if we limit the number of agents to two or restrict the costs to be identical.

\subsection{Two Agents with Additive Valuations}
\label{subsec:inapprox_two_agents}

As \Cref{lem:maxpe} will show, $\maxpe$ is NP-hard even for a single agent with additive valuation. The following theorem uses \Cref{lem:maxpe} to demonstrate a strong inapproximability of UWO-WP in the additive setting, even in the very restricted case with two agents.

\begin{theorem}\label{thm:inapprox}
Let $f(m)$ be a polynomial time computable function. Even for two agents, there is no polynomial time $f(m)$-approximation algorithm for the UWO-WP problem, assuming  $P \neq  NP$.
\end{theorem} 
\begin{proof}
Suppose, on the contrary, that a polynomial time algorithm exists with  approximation guarantee $f(m)$ for UWO-WP with $n=2$. By the argument used to prove \Cref{thm:gap_reduction}, we see that such an algorithm can be used to solve $\maxpe$ for one agent in polynomial time. However, as we will show in \Cref{lem:maxpe}, $\maxpe$ is {NP-hard} even for one agent, contradicting $P \neq  NP$.
\end{proof}

\begin{lemma}
\label{lem:maxpe}
It is NP-hard to decide $\maxpe\geq t$, even for a single agent with additive valuation. 
\end{lemma}
\begin{proof}
We give a reduction from the NP-complete $\operatorname{PARTITION}$ problem: given integers $a_1,\cdots, a_m$, the problem asks if there exists a subset $S$ such that $\sum_{i\in S} a_i = \frac{1}{2} \sum_{i\in M} a_i $? For ease of notation, we denote $\gamma=\frac{1}{2} \sum_{i\in M} a_i $. Given any instance of the partition problem, we construct an instance of $\maxpe\geq t$ with a single agent such that $v_{1i}= a_i$ for each $i\in M$, $b_1 =\gamma$, $C_i=\frac{a_i}{2} $, and $t=\frac{\gamma}{2}$.

The $\maxpe\geq t$ decision problem for this instance is then
\begin{align*}
\max_{W\subseteq M} \ \  \min\left(\gamma, \sum_{i\in W} a_i \right)  - \sum_{i\in W} \frac{a_i}{2}  \geq t
\end{align*}
If there exists a set $S^*$ such that $\sum_{i\in S^*} a_i = \gamma$, then the corresponding $\maxpe\geq t$ instance is a $\mathtt{Yes}$ instance. This holds since 
\begin{align*}
\min\left(\gamma, \sum_{i\in S^*} a_i \right)  - \sum_{i\in S*} \frac{a_i}{2}  \ = \  \gamma - \frac{\gamma}{2}  =\frac{\gamma}{2}  =t
\end{align*}
On the other hand, if for any $S\subseteq\{1,\cdots, m\}$ we have $\sum_{i\in S} a_i\not = \gamma$, then the corresponding instance of $\maxpe\geq t$ is a $\mathtt{No}$ instance. This can be seen by the following observation:
\begin{itemize}
\item If $\sum_{i\in S} a_i<\gamma$, then $\min(\gamma,\sum_{i\in S} a_i)  - \sum_{i\in S} \frac{a_i}{2}= \frac{1}{2} \sum_{i\in S} a_i  <  \frac{\gamma}{2}  =t$
\item If $\sum_{i\in S} a_i>\gamma$, then $\min(\gamma,\sum_{i\in S} a_i) - \sum_{i\in S} \frac{a_i}{2}= \gamma-\frac{1}{2} \sum_{i\in S} a_i < \frac{\gamma}{2}  = t $ 
\end{itemize}
Thus, we see that a solution to the  $\operatorname{PARTITION}$ problem exists if and only if the corresponding $\maxpe\geq t$ instance has a valid solution.
\end{proof}

\begin{remark}
  The inapproximability given by \Cref{thm:inapprox} holds for a broader class of social welfare functions that includes (but is not limited to) Nash social welfare and egalitarian social welfare, in addition to utilitarian.
\end{remark}

Even though UWO admits an FPTAS, with the inclusion of the simple and intuitive constraint of WP, we obtain the above strong inapproximability results. We next show that the inapproximability holds even if all projects costs are identical.

\subsection{Identical Costs Setting}
\label{subsec:inapprox_identical_costs}

We consider the special case in which all projects have identical costs.

In this setting we have $C_j = C$, $\forall j\in M$, for some constant $C$. We assume $C=1$ for simplicity. We first note that the UWO allocation in the identical costs setting is computable in polynomial time.\footnote{Specifically, we can sort the projects according to their social welfare and greedily select them until there is no remaining project with non-negative social welfare or until no budget remains.}

Given UWO is polynomial time computable in the identical costs setting, and the same problem is NP-hard in the general setting, we may take this as an encouraging sign that a tractable algorithm for UWO-WP may exist in our restricted setting. However, as we will show, UWO-WP does not admit any polynomial approximation guarantees, even in this restricted setting. Our argument follows similarly to that used in the general setting. Intuitively, to be able to afford any single project in our setting, we must be able to identify a subset of projects that extracts a single dollar of excess payment when such a subset exists. We define the decision version of the MaxPE problem with $t=1$ ($\maxpe\geq1$) as $\max_{W\subseteq M}  \sum_{i\in N} \min(b_i,v_i(W) )-\lvert W \rvert \geq 1$.

\begin{restatable}{lemma}{lemmaxpeone} 
\label{lem:maxpe1}
    It is NP-hard to decide $\maxpe\geq1$, even with identical budgets.
\end{restatable}

With the preceding lemma, the following inapproximability result follows from an analogous argument to that used in the proof of \Cref{thm:inapprox}. 

\begin{theorem}
    Let $f(n,m)$ be a polynomial time computable function. Unless $P\neq NP$, there is no polynomial time $f(n,m)$-approximation algorithm for UWO-WP problem in the identical costs setting.
\end{theorem}

Having found a strong inapproximability in various settings with additive valuations, we turn to another well-motivated setting and give an approximation algorithm for UWO-WP.

\section{An Algorithm for Laminar Single-minded Valuations}
\label{sec:laminar_sm}
In this section, we study a setting in which each agent derives non-zero value from a project bundle if and only if the bundle contains the set of projects desired by that agent, which we refer to as the agent's \textit{demand set}. We refer to this as the setting with \textit{single-minded} valuations, following the precedent of single-minded bidders in the combinatorial auctions literature \citep{APT+04a,CDS04a}. Single-minded agents are also well-studied in the mechanism design literature \citep{DGS+20a,Aziz20b} and fair division literature \citep{BLM16a}. In contrast to the additive setting, the single-minded setting captures project complementarities in agent valuations. 
We now define the single-minded setting formally.
\begin{definition}[Single-minded Valuations]
    Valuations $\mathbf{v} = (v_i)_{i\in N}$ are \textit{single-minded} if, for each $i\in N$, there exists $D_i\subseteq M$ and $z_i\in \mathbb{R}_{\ge 0}$ such that $v_i(W) = z_i$ if $D_i\subseteq W$ and zero otherwise. We refer to $D_i$ as agent $i$'s \textit{demand set}.
\end{definition}

First note that for any set of agents with the same demand set, i.e. $D_i=T$ for all $i\in N^{\prime}\subseteq N$ and some $T\subseteq M$, the problem is equivalent to the problem where $N^{\prime}$ is replaced by a single agent $i^{\prime}$ with $b_{i^{\prime}} = \sum_{i\in N^{\prime}} b_i$ and $v_{i^{\prime}}(W) = \sum_{i\in N^{\prime}} v_i(W)$ for all $W\subseteq M$. After executing this step for all such sets of agents, we have an equivalent instance with distinct demand sets, $\mathbf{D}$. Henceforth, we assume without loss of generality that all demand sets are distinct. We will denote by $N_{S}$ the set of agents whose demand sets are contained in the set $S\subseteq M$, i.e. $N_S=\{i\in N\vert D_i\subseteq S\}$. 
We refer to any set which maximizes the $\pe$ quantity, as defined in \Cref{subsec:gap_intro}, as a \textit{$\maxpe$ set}.

In this section, we give an FPTAS for a particular restriction of the single-minded setting, that with \textit{laminar single-minded valuations}, which we will define now.

\begin{definition}[Laminar Single-minded Valuations]
    Valuations $\mathbf{v} = (v_i)_{i\in N}$ are \textit{laminar single-minded} if they are single-minded and the demand set family $\mathbf{D}=(D_i)_{i\in N}$ is a laminar set family, i.e. for every $D_i, D_j \in \mathbf{D}$, $D_i\cap D_j$ is either empty or equal to $D_i$ or $D_j$.
\end{definition}

The laminar single-minded case describes instances in which the project set is naturally partitioned into categories and sub-categories. Relating back to \Cref{ex:towns} where neighboring towns fund shared public projects, we might consider a scenario in which each town's budget is constrained to be used toward certain project types, such as infrastructure, housing, or green space. Further, the towns constrained to spend on green space might be divided by whether they want a park with athletic fields or walking trails. In this case, the towns' valuations would comprise a laminar set family.

UWO remains NP-hard in our restricted setting with laminar demand sets. 
This hardness results follows from a straightforward reduction from $\operatorname{KNAPSACK}$ to an instance of our problem with disjoint (and thus, laminar) demand sets which maps knapsack items to agent demand sets. 
If we set $b_i = z_i$ for all $i\in N$, we have that the UWO-WP outcome will also be UWO and thus the same reduction shows NP-hardness of UWO-WP in the single-minded laminar setting.

Given the hardness results, we search for an approximation of UWO-WP in the laminar single-minded setting.
The following is the main result of the section, which provides an FPTAS for UWO-WP in the single-minded setting with laminar demand sets.

\begin{theorem}\label{thm:uwowp_laminar}
    In the laminar single-minded setting, the UWO-WP problem admits an FPTAS.
\end{theorem}

Before we prove this result, we state three key lemmas used in the proof. We first point out that $\maxpe$, despite being NP-hard for a single agent in the additive setting, is polynomial-time solvable in the single-minded setting. 

\begin{restatable}{lemma}{lemmaxpesm}
\label{lem:mpe_sm_general}
    The $\maxpe$ problem can be solved in polynomial time in the single-minded setting.
\end{restatable}

The following lemma, which states that every MaxPE set is contained in at least one UWO-WP outcome, also holds for the single-minded setting in general.

\begin{lemma}\label{lem:mpe_in_uwowp}
    In the single-minded setting, for any MaxPE set $Q$, there exists a UWO-WP outcome $W^{\star}$ such that $W^{\star}\supseteq Q$.
\end{lemma}

\begin{proof}
  Let $Q$ be any MaxPE set and let $W$ be any UWO-WP outcome. If $Q\subseteq W$, we are done. Assume $Q\not \subseteq W$. Then, there is a non-empty set $Q^{\prime}=Q\setminus W$. We will now show that $W^{\prime} = W \cup Q^{\prime}$ is UWO-WP, which implies that $Q$ is a subset of a UWO-WP outcome and concludes our proof.

  First note that because $Q$ is a MaxPE set, we have that $\pe(Q) \geq \pe(Q\setminus Q^{\prime})$, which implies that
  \begin{align}
      & \sum_{i\in N_Q} \min(b_i, z_i) - C(Q) \geq \sum_{\crampedclap{i\in N_{Q\setminus Q'}}} \min(b_i, z_i) - \sum_{\crampedclap{j\in Q\setminus Q^{\prime}}} C_j \nonumber \\
      \implies & \sum_{i\in N_Q} \min(b_i, z_i) - \sum_{\crampedclap{\{i\in N_Q | D_i\cap Q^{\prime}=\emptyset\}}} \min(b_i, z_i) - C(Q^{\prime}) \geq 0 \nonumber \\
      \implies & \sum_{\crampedclap{\{i\in N_Q | D_i\cap Q^{\prime}\neq \emptyset\}}} \min(b_i, z_i) - C(Q^{\prime}) \geq 0 
      \label{eq:qprime_wp}
  \end{align}
    
  We now show that $W^{\prime}$ is feasible and WP. Because $W$ and $Q^{\prime}$ are disjoint and $W\supseteq Q\setminus Q^{\prime}$, observe that
 
  \begin{align*}
      & \sum_{\crampedclap{i\in N_{W'}}} \min(b_i, z_i) - C(W^{\prime}) \\
      = & \sum_{\crampedclap{\{i\in N_{W'}| D_i\cap Q^{\prime}=\emptyset\}}} \min(b_i, z_i) - C(W) \quad + \quad \sum_{\crampedclap{\{i\in N_{W'} | D_i\cap Q^{\prime}\neq \emptyset\}}} \min(b_i, z_i) - C(Q^{\prime}) \\
      \geq & \sum_{i\in N_W} \min(b_i, z_i) - C(W) + \sum_{\crampedclap{\{i\in N_Q | D_i\cap Q^{\prime}\neq \emptyset\}}} \min(b_i, z_i) - C(Q^{\prime}) \geq 0
  \end{align*}
  
  The final inequality follows from \Cref{eq:qprime_wp} in conjunction with $W$ being WP. All that remains is to show that the addition of $Q^{\prime}$ to $W$ weakly improves social welfare. By a similar argument as above, it follows that $\SW(W^{\prime})$ equals
  $$
  \left[ \sum_{i\in N_W} z_i - C(W) \right] + \left[\hspace{2.5em} \sum_{\crampedclap{\{i\in N_Q | D_i\cap Q^{\prime}\neq \emptyset\}}} z_i - C(Q^{\prime}) \right]
  \geq \SW (W) + \sum_{\crampedclap{\{i\in N_Q | D_i\cap Q^{\prime}\neq \emptyset\}}} \min(b_i, z_i) - C(Q^{\prime}) \geq \SW (W) .
  $$
\end{proof}

We construct an FPTAS in the laminar single-minded setting, by providing an approximation-preserving reduction \citep{Vazi01a} to another generalization of the classical knapsack problem called \textit{knapsack on conflict graphs} (KCG) \citep{PfSc09a}, which we now describe. 

An instance of KCG is defined by a knapsack capacity $B$ and an undirected graph $G=(V,E)$ referred to as a \textit{conflict graph}, where each node $i\in V$ has an associated profit $p_i$ and weight $w_i$. The KCG problem is then expressed by the following optimization problem:

\begin{equation}
\begin{aligned}
    \max_{S\subseteq V} \quad & \displaystyle\sum_{i\in S} p_i \nonumber\\
    \textrm{subject to} \quad & \sum_{i\in S} w_i \leq B \nonumber \\
    & \lvert \{(i,j) \in E | i,j\in S\} \rvert = 0
\end{aligned}
\end{equation}

The following lemma states the conflict graph defined by a laminar set family is \textit{chordal}, i.e. all cycles of four or more vertices have a chord, which is an edge that is not part of the cycle but connects two vertices of the cycle.

\begin{lemma}
\label{lem:chordal}
    Given a laminar set family $\mathbf{D}=\{D_1,\cdots,D_n\}$, the graph $G=(V,E)$ where $V=[n]$ and $E=\{(i,j) | D_i\subseteq D_j\}$ is \textit{chordal}.
\end{lemma}
\begin{proof}
    Consider such a graph $G=(V,E)$. Let $E^c\subseteq E$ be a cycle of length four or more. Pick any edge $(i,j)\in E^c$ and assume without loss of generality that $D_i\subseteq D_j$. Now pick any edge $(i,k)\in E^c$, which must exist because $E^c$ is a cycle. We have that $D_i\cap D_k\neq \emptyset$ by the construction of the graph and thus $D_k\cap D_j \neq \emptyset$. Since $\mathbf{D}$ is a laminar set family, this means that $D_j$ and $D_k$ form a subset relation, which implies $(j,k)\in E$. $(j,k)$ is a chord since it connects two vertices in $E^c$ and cannot be in $E^c$ since it would form a triangle, contradicting $E^c$ being a cycle of length four or more. 
\end{proof}

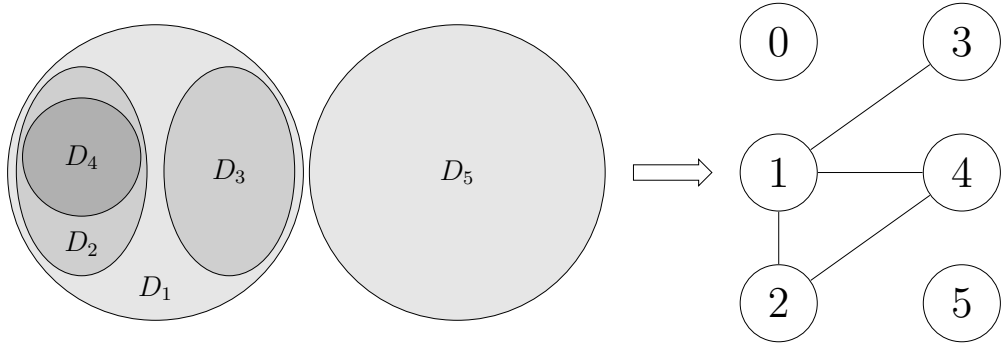
\begin{figure}
\begin{FitToWidth}[0.8\columnwidth]
\begin{tikzpicture}[roundnode/.style={draw,circle,fill,fill opacity=0.1},
          elipnode/.style={draw,ellipse,fill,fill opacity=0.1},
          baseline=(current bounding box.center)]
    \def\scale{1.3}
    \def\prop{0.1}
    
    \node[roundnode,
        inner sep=20*\prop cm,
        label={[above = 2*\prop cm,name=d1_lab, scale=\scale]south:$D_1$}] (d1) {};
    
    \node[elipnode,
        minimum height=40*\prop cm,
        minimum width=25*\prop cm,
        above left= 4*\prop cm and 0cm of d1_lab,
        label={[above=2*\prop cm, name=d2_lab, scale=\scale]south:$D_2$}] (d2) {};
    
    \node[elipnode,
        minimum height=40*\prop cm,
        minimum width=25*\prop cm,
        above right= 4*\prop cm and 0 cm of d1_lab,
        label={[name=d3_lab, scale=\scale]center:$D_3$}] {};
    
    \node[roundnode,
        inner sep = 8*\prop cm,
        above = \prop cm of d2_lab,
        fill opacity=0.15, 
          text opacity=1, 
        label={[scale=\scale]center:$D_4$}] {};
    
    \node[roundnode,
        inner sep=20*\prop cm,
        right= 0.1cm of d1,
        label={[scale=\scale]center:$D_5$}] (d5){};
\end{tikzpicture}

\quad
\begin{tikzpicture}[baseline=(current bounding box.center)]
  \node[single arrow, draw=black,
      minimum width = 10pt, single arrow head extend=3pt,
      minimum height=1.5cm] {}; 
  \end{tikzpicture}
\quad

\begin{tikzpicture}[roundnode/.style={draw,circle,scale=2},
          baseline=(current bounding box.center), every label/.style={text=black}]
  
  \node[roundnode] (3) {$3$};
  \node[roundnode,below=of 3] (4) {$4$};
  \node[roundnode,below=of 4] (5) {$5$};
  \node[roundnode, left= 2cm of 3] (0) {$0$};
  \node[roundnode, left=  2cm of 4] (1) {$1$};
  \node[roundnode, left= 2cm of 5] (2) {$2$};

  \draw (3) -- (1) -- (2) -- (4) -- (1);
\end{tikzpicture}
\end{FitToWidth}

\caption{The mapping from an instance of the laminar single-minded UWO-WP problem to an instance of KCG as described in the proof of \Cref{thm:uwowp_laminar}.}
\label{fig:kcg_reduction}
\end{figure}

\begin{proof}[Proof of \Cref{thm:uwowp_laminar}]
    Let $I = \langle N, M, \mathbf{v}, \mathbf{b}, \mathbf{C} \rangle$ be an instance of UWO-WP problem with  laminar single-minded demand sets. Let $Q$ be a MaxPE set computed by the polytime algorithm given in the proof of \Cref{lem:mpe_sm_general} and let $N^\prime$ be the set of agents whose demand sets are not in $Q$, i.e. $N^\prime = \{i\in N | D_i\nsubseteq Q\}$. 
    
    We know that a UWO-WP solution which contains $Q$ is guaranteed to exist by \Cref{lem:mpe_in_uwowp}. We proceed by reducing the problem of computing such a UWO-WP solution $W^\star \supseteq Q$ to an instance of KCG as follows:
    \begin{itemize}
        \item $V = \{0\}\cup N$
        \item $p_i = 
          \begin{cases}
            \hspace{2.5em} \displaystyle\sum_{\crampedclap{j\in N\setminus N^\prime}} z_j - C(Q) & \quad \text{if}\ i=0 \\
            \hspace{2.5em} \displaystyle\sum_{\crampedclap{j\in N^\prime | D_j\subseteq D_i\cup Q}} z_j - C(D_i\setminus Q) & \quad \text{otherwise}
          \end{cases} $
        \item $w_i = 
          \begin{cases}
            0 & \text{if}\ i=0 \\
            C(D_i\setminus Q) - \displaystyle\sum_{\crampedclap{\{j\in N^\prime | D_j\subseteq D_i\}}} \min(b_j, z_j) & \text{otherwise}
          \end{cases}$
        \item $B = \pe(Q)$
        \item $E = \{(i,j) \vert D_i\subseteq D_j, \forall i,j\in N\}$
    \end{itemize}

    See \Cref{fig:kcg_reduction} for an example of the mapping from an instance of UWO-WP to an instance of KCG. In words, each node in the conflict graph besides $\{0\}$ corresponds to an agent in the UWO-WP instance $I$. The knapsack capacity corresponds to the slack in our feasibility constraint produced by the selection of the MaxPE set. Lastly, node $\{0\}$ corresponds to the agents whose demand sets are in the MaxPE outcome, i.e. $N\setminus N^\prime$. Thus, since we seek a UWO-WP solution containing $Q$, this node provides profit equal to $\SW(Q)$ for zero cost and is not connected by an edge to any other node.

    Let $S$ be a feasible solution to the instance of KCG reduced from $I$ and denote $S^0=S\setminus \{0\}$. We posit that $W = Q\cup \bigcup_{i\in S^0} D_i$ satisfies feasibility and WP. Since $S$ is a feasible KCG solution, we know $\sum_{i\in S} w_i = \sum_{i\in S^0} w_i \leq B$, which we can write as
    
    $$\sum_{\crampedclap{i\in S^0}} \left[ C(D_i\setminus Q) - \sum_{\crampedclap{\{j\in N^\prime \vert D_j\subseteq D_i\}}} \min(b_j, z_j) \right]  \leq \sum_{\crampedclap{i\in N_Q}} \min(b_i, z_i) - C(Q).$$

    \indent Since $S$ is feasible, we know that no two members of $\{D_i\}_{i\in S}$ can form a subset relation and thus all members are disjoint by the laminarity of $\mathbf{D}$. From this, we can also conclude that $\{j\in N^\prime \vert D_j\subseteq D_{i_1}\} \cap \{j\in N^\prime \vert D_j\subseteq D_{i_2}\} = \emptyset$ for all $i_1, i_2\in S$. Thus, the above inequality becomes

    \begin{align*}
      &\sum_{j\in \bigcup\limits_{\crampedclap{i\in S^0}} D_i \setminus Q} C_j - \quad \sum_{\crampedclap{\{j\in N^\prime \vert D_j\subseteq \bigcup\limits_{\crampedclap{i\in S^0}} D_i \}}} \min(b_j,z_j) \leq \sum_{\crampedclap{i\in N_Q}} \min(b_i,z_i) - C(Q)  \\
      &\implies \sum_{\crampedclap{i\in N_Q}} \min(b_i,z_i) + \quad \sum_{\crampedclap{\{i\in N\setminus N_Q \vert D_i\subseteq \bigcup\limits_{\crampedclap{j\in S^0}} D_j\}}} \hspace{0.5em} \min(b_i,z_i) - \sum_{\crampedclap{j\in Q\cup \bigcup\limits_{\crampedclap{i\in S^0}} D_i}} C_j \geq 0 \\ 
      &\implies \sum_{\crampedclap{\{i\in N\vert D_i\subseteq Q\cup \bigcup\limits_{\crampedclap{i\in S^0}} D_i\}}} \min(b_i,z_i) - \sum_{\crampedclap{j\in Q\cup \bigcup\limits_{\crampedclap{i\in S^0}} D_i}} C_j \geq 0
    \end{align*}
    which means $W$ is feasible and WP.

    All that remains to complete our approximation-preserving reduction is to equate the objective values of corresponding solutions of our original and reduced instances. Since node $\{0\}$ has weight $0$ and non-negative profit, we can restrict our attention to KCG solutions containing $\{0\}$ without loss of generality. We now show that the objective value of any feasible KCG solution, $S$, is equal to that of the resulting project bundle, $W$:
    \begin{align}
        \sum_{i\in S} p_i & = \sum_{\crampedclap{i\in S^0}} \left[ \hspace{2.5em} \sum_{\crampedclap{\{j\in N^\prime \vert D_j\subseteq D_i\cup Q\}}} z_j - C(D_i\setminus Q) \right] + \sum_{\crampedclap{j\in N\setminus N^\prime}} z_j - C(Q) \nonumber \\
        & =\quad \sum_{\crampedclap{\{i\in N^\prime \vert D_i\subseteq Q\cup \bigcup\limits_{\crampedclap{j\in S^0}} D_j\}}} \quad z_i \quad + \quad \displaystyle\sum_{\crampedclap{i\in N_Q}} z_i \quad - \quad \sum_{\crampedclap{j\in \bigcup\limits_{\crampedclap{i\in S^0}} D_i\setminus Q}} C_j - C(Q) \nonumber \\
        & =\quad \sum_{\crampedclap{i\in N_W}} z_i - C(W) \quad = \quad \SW(W) \nonumber
    \end{align}
    
    It is apparent that the reduction runs in polynomial time. \citet{PfSc09a} provide an FPTAS for KCG on chordal graphs. By \Cref{lem:chordal}, 
    $G$ is a chordal graph. Thus, laminar single-minded UWO-WP admits an FPTAS. 
\end{proof}

In the single-minded setting, the problem of welfare maximization resembles the well-studied \textit{set union knapsack problem} (SUKP), which is hard to approximate in the general case \citep{GNY94a,Nehm95a,Arul14a}. 
Despite SUKP's resemblance to our problems of interest, approximation results from SUKP or its restricted cases, such as those described by \citet{GNY94a}, do not imply any of our results. However, an argument similar to that given to prove \Cref{thm:uwowp_laminar} reveals a new restriction under which SUKP admits an FPTAS, namely the restriction to laminar item sets. Since this problem is not the topic of this work, we leave the theorem and proof of this finding, along with a formal description of SUKP, to \Cref{app:sukp}.

\section{An Algorithm for Symmetric Valuations}
\label{sec:symm}
{We now focus our attention on the domain of symmetric valuations. We say that a valuation function $v_i$ is \textit{symmetric} if it depends only on the size of the input, i.e $v_i(S) = v_i(T)$ whenever $|S| = |T|$. Such valuations describe a setting where multiple units of the same project or item are available and agents' valuations are solely dependent on the quantity of items selected.} Note that symmetric valuations are not necessarily additive. For instance, the valuations can be $v_i(S)=\sqrt{|S|}$ or $v_i(S)=e^{|S|}-1$, which are submodular and supermodular valuations, respectively, and neither of which are additive.

Symmetric valuations are standard restrictions in the mechanism design literature. For example, \citet{MhKt15} use symmetric valuations to capture "a spectrum auction of equally sized channels which all offer very similar conditions." Other research has investigated cost sharing mechanisms for symmetric valuations \citep{birmpas2019cost}.

\begin{algorithm}[!t]
  \caption{Symmetric Valuation}\label{alg:symm}
  \begin{algorithmic}
  \State \textbf{Input:}  $\langle N, M, \mathbf{v}, \pmb{b}, \mathbf{C} \rangle$ where valuations are symmetric
  \State Order projects by increasing cost, i.e $C_1\leq C_2 \leq \cdots \leq  C_m $
  \State $\Omega = \emptyset$
  \For{$k=1$ to $m$ }
  \If{$\sum_{i\in N}\min\{v_i([k]), b_i\} \geq C([k])$}
      \State $\Omega \leftarrow \Omega \cup k$
  \EndIf
  \EndFor
  \State $k^*\leftarrow \arg\max \limits_{k\in \Omega} \SW ([k])$
  \For{$i=1$ to $n$}
    \If{$\sum_{j=1}^i \min\{v_j([k^*]), b_j\} < C(W)$}
        \State $x_i \gets \min\{v_i([k^*]), b_i\}$
    \Else{}
        \State $x_i \gets C([k^*]) - \sum_{j=1}^{i-1} \min\{v_j([k^*]), b_j\}$
    \EndIf
  \EndFor
  \State \textbf{return} $W=[k^*] $ and  $\textbf{x}=(x_1,\cdots, x_n)$ 
  \end{algorithmic}
\end{algorithm}

\begin{theorem}
For symmetric valuations, Algorithm \ref{alg:symm} finds UWO-WP in polynomial time.
\end{theorem}
\begin{proof}
  As outlined in Algorithm \ref{alg:symm}, we assume the projects are ordered by increasing cost. We denote the set of projects $\{1,...,k\}$ by $[k]$ for simplicity of notation. Let $W^*$ be the solution to the UWO-WP problem with $\ell = |W^*|$, and $W$ be the set of projects returned by Algorithm~\ref{alg:symm}. Note that $v_i(W^*)=v_i([\ell])$ for each $i\in N$, since the valuations are symmetric. We see that $[\ell]$ can be funded in a WP and WBB manner since $W^*$ can be funded in such a manner and the cost of $W^*$ is weakly greater than that of $[\ell]$. From Lemma~\ref{lem: WP-WBB}, it follows that $\sum_{i\in N}\min\{v_i([\ell]), b_i\} \geq C([\ell])$, and hence $\ell\in \Omega$. Thus we have,
  \begin{align*}
  \SW(W) &\geq \SW([\ell])\\
   & = \sum_{i\in N} v_i([\ell]) - C([\ell]) = \sum_{i\in N} v_i(W^*) - C([\ell]) \\ 
  &\geq  \sum_{i\in N} v_i(W^*) - C(W^*) = \SW(W^*)
  \end{align*}
  where the first inequality follows since $\ell\in \Omega$ and $W=[k^*]$ with $k^* = \arg\max_{k\in \Omega} \SW ([k])$. Also note that for each $k\in \Omega$, the set of projects $[k]$ can be funded in a WP and WBB manner by Lemma~\ref{lem: WP-WBB}, and thus $W$ can also be funded in such a way. Hence, $W$ is indeed a UWO-WP solution.
\end{proof}

\section{Experiments}
\label{sec:experiments}

\begin{figure}[!t]
    \centering
    \includegraphics[width=\columnwidth]{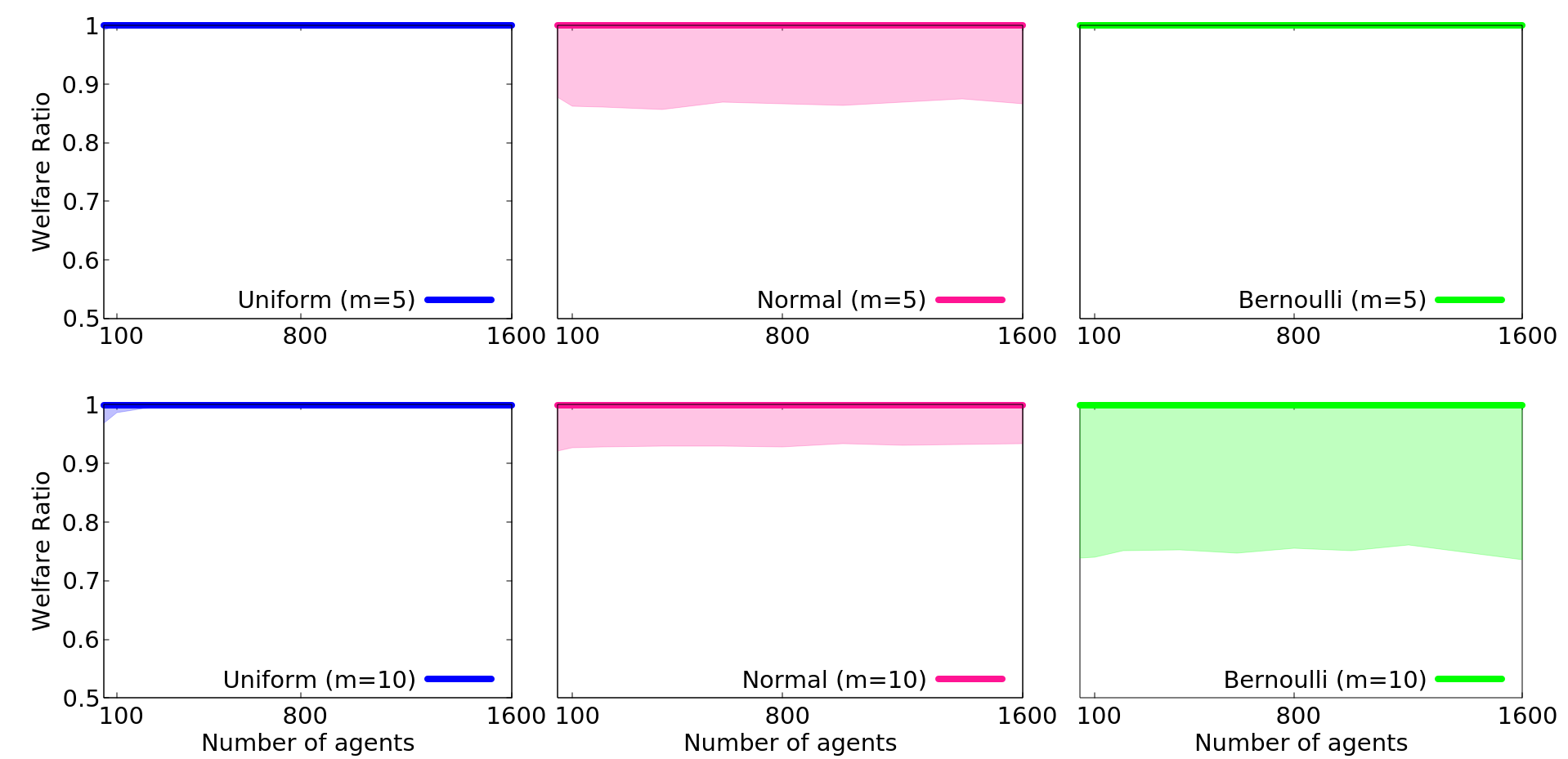}
    \caption{Welfare Ratio $\SW(greedy)/\SW(\text{UWO-WP})$ for $m=5$ (top) and $m=10$ (bottom). Solid line is the median welfare ratio across 10,000 instances; shaded region indicates the range of welfare ratio for 90\% of the instances}
    \label{fig:expt}
\end{figure}

We know that UWO-WP is not just hard but inapproximable for general additive valuations. In this section, we use real PB elections and synthetically generated data to construct simulations evaluating how well a simple greedy algorithm approximates UWO-WP welfare.

In the algorithm we henceforth refer to as $greedy$, we sort the projects in descending order of bang per buck, i.e. welfare obtained per cost. We pick a new project in this order, provided it does not violate WP or WBB, and continue in this fashion until no such project remains. For each of our simulation instances, we evaluate $greedy$'s performance by calculating the \textit{welfare ratio} between the greedy solution and the optimal solution, which is computed through brute force search, i.e. $\frac{\SW(greedy)}{\SW(\text{UWO-WP})}$.

\subsection*{Pabulib Experiments}
We begin by constructing instances of our problem using real PB election data from Pabulib \citep{SST20a}. We convert all approval elections in the dataset to instances of our problem, restricting only to those elections with $20$ projects or fewer due to computational constraints associated with calculating the optimal solution. To construct the instance, we distribute the election's budget evenly among all voters and assume that agents derive non-zero utility from a project if and only if they approve of that project in their ballot. We then scale all agents' utilities up so that the sum of all agents' utilities for the grand bundle is equal to the total cost of the grand bundle. We choose this scaling to avoid creating instances with trivial optima. In the absence of more detailed valuation information, we consider this assumption fair as trivial instances would be unlikely to warrant the initiation of a PB process to begin with.

For the $163$ elections meeting the criteria laid out above, we find that the $greedy$ algorithm achieves a welfare ratio of greater than $0.98$ in $50\%$ of elections and achieves a welfare ratio of greater than $0.75$ in $90\%$ of elections.

\subsection*{Synthetic Experiments}
Having found that $greedy$ performs quite well on data from real PB elections, we set out to test $greedy$ on a synthetic dataset with more varied input assumptions and a much larger sample size. The instances we use in our synthetic simulations are of the three following types:
\begin{enumerate}
    \item For \textit{uniform} instances, we sample agent valuations from a Uniform distribution, $v_{ij} \sim \mathcal{U}[0,1]$.
    \item For \textit{normal} instances, we uniformly sample a mean and variance for each project to define a Normal distribution from which all agents' valuations are drawn for that project. That is, for each $j\in M$, we sample $\mu_j \sim \mathcal{U}[0,1]$ and $\sigma_j \sim \mathcal{U}[0,0.5]$, and draw $v_{ij}\sim \mathcal{N}(\mu_j, \sigma_j)$, shifting to ensure positive valuations.
    \item For \textit{bernoulli} instances, for each project $j$, we first uniformly sample a probability, $p_j \sim \mathcal{U}[0,1]$, and a constant, $\nu_j \sim \mathcal{U}[0,1]$. We then draw agents' valuations from a Bernoulli distribution with success probability $p_j$ and multiply the result by $\nu_j$, i.e. $v_{ij}\sim \nu_j \cdot {\rm Bernoulli}(p_j)$.
\end{enumerate} 

For all instances, project costs are sampled uniformly, $C_j \sim \mathcal{U}[0.75\cdot \sum_i v_{ij}, \sum_i v_{ij} ]$. By ensuring costs are close to the sum of agents' valuations for each project, we make WP more difficult to satisfy and avoid trivial instances. The total budget is set to half of the total of project costs and individual budgets are fractions of this overall budget, drawn uniformly at random and normalized to sum to 1. The \textit{uniform} instance family captures a setting where agents' valuations are perfectly random and not correlated with one another. Since this is not the case in most realistic settings, we use the \textit{normal} instance type to capture scenarios in which there is a degree of consensus amongst all agents on the value of each project. Lastly, we study the \textit{bernoulli} instance type, which captures approval-based utilities, a modeling assumption well-studied in PB \citep{ALT18a, NtPf19}.

For each combination of instance family, $m = \{5, 10\}$, and $n$ ranging from 10 to 1,600, we generate 10,000 synthetic instances. The results are reported in Figure~\ref{fig:expt} for each value of $n$ and $m$, where the solid line is the median welfare ratio across all 10,000 instances and the shaded region below outlines the $10^{th}$ percentile welfare ratio across the same instances. Figure \ref{fig:expt} shows that $greedy$ obtains optimal welfare (i.e. UWO-WP) in at least $50\%$ of instances for each instance type and number of agents and projects. The shaded region of the charts shows that the $10^{th}$ percentile performance of $greedy$ drops significantly below the optimum for \textit{normal} instances with $m=5$ and \textit{bernoulli} instances with $m=10$. Nevertheless, even for these instances, we point out that $greedy$ achieves at least $70\%$ of optimal welfare in $90\%$ of instances. 

Altogether, we observe that the $greedy$ algorithm achieves close-to-optimal performance in the vast majority of a wide variety of instances with reasonable input assumptions. Such performance from an algorithm as unsophisticated as $greedy$ comes as a notable surprise given our strong inapproximability results.

\section{Conclusion}
In this paper, we laid the groundwork for an important and general PB framework which captures resource pooling. We focused on welfare maximization subject to minimal participation axioms that agents would expect the mechanism to satisfy and proved several hardness and inapproximability results. We also identified and gave explicit algorithms for two natural classes of instances. 

Our work leaves interesting open problems, and provides a framework for further exploration of PB with resource pooling. Specifically, while we showed an FPTAS for UWO-WP in the laminar single-minded setting, it remains open whether there exists a bounded approximation algorithm for UWO-WP in the general single-minded case. Furthermore, future work can focus on other restrictions which bypass our inapproximability. Lastly, while our participation axiom captures some notion of fairness, it would be interesting to explore other fairness concepts within our framework.\\

\noindent \textbf{Acknowledgements} \\
\noindent Mashbat Suzuki was partially supported by the ARC Laureate Project FL200100204 on "Trustworthy AI".

\bibliographystyle{ACM-Reference-Format.bst}
\bibliography{ref,PBC-haris_master,PBC-aziz_personal}

\clearpage
\appendix

\input{appendix.tex}

\end{document}

%% file: appendix.tex

\section{Proof of Remark~\ref{rmk:uwo_fptas}}  \label{app:uwo_fptas}
\rmkuwo*
\begin{proof}
  To show NP-hardness of UWO with a single agent with additive valuation, we first reduce the problem from the NP-complete $\operatorname{KNAPSACK}$ problem: given $m$ items with weights $w_1,\cdots, w_m$ and values $s_1,\cdots, s_m$, capacity $B$ and value $V$, does there exist a subset $F\subseteq \{1,\cdots,m\}$ such that $\sum_{j\in F} w_j \leq B$ and $\sum_{j\in F} s_j \geq V$? Given an instance of the  $\operatorname{KNAPSACK}$ problem, we build an instance of the $\operatorname{UWO}$ problem as follows:
\begin{itemize}
  \item The set of projects is the set of items. Each project has cost $C_j=w_j$.
  \item The budget of the agent is $B$.
  \item The value of each project to the agent is $v_j = s_j + w_j$.
\end{itemize}
We can see that the social welfare obtained by choosing a set of projects $F$ is exactly equal to the value of choosing set of items $F$ in the $\operatorname{KNAPSACK}$ problem i.e, $\SW(F)=\sum_{j\in F}v_j-C_j=\sum_{j\in F} s_j $. Also note  that the WBB condition is satisfied if and only if the capacity constraint is satisfied. It follows that a solution with value at least $V$ exists in the $\operatorname{KNAPSACK}$ problem if and only if there exists a set of projects whose social welfare in the corresponding UWO instance is at least $V$. The NP-hardness follows immediately \citep{GaJo79a}.

We prove existence of an FPTAS for UWO via an approximation-preserving reduction to $\operatorname{KNAPSACK}$. Given an instance of UWO, we build an instance of $\operatorname{KNAPSACK}$ as follows:
\begin{itemize}
  \item The set of items is the set of projects. Each item has weight $w_j = C_j$.
  \item The knapsack capacity is the sum of agent budgets, i.e. $B = \sum_{i\in N} b_i$.
  \item The value of each item is $s_j = \sum_{i\in N} v_{ij} - C_j$.
\end{itemize}
Again, we note that the social welfare obtained by choosing a set of items $F$ is exactly equal to the value of choosing a set of projects $F$ in the $\operatorname{KNAPSACK}$ problem, i.e. 
  \begin{align}
    \sum_{j\in F} s_j & = \sum_{j\in F} \left( \sum_{i\in N} v_{ij} - C_j \right) \nonumber \\
    & =  \sum_{i\in N} \sum_{j\in F} v_{ij} - \sum_{j\in F} C_j \nonumber \\
    & = \sum_{i\in N} v_i(F) - \sum_{j\in F} C_j \nonumber \\
    & = SW(F) \nonumber
  \end{align}
It is clear to see that the knapsack capacity constraint is satisfied if and only if the corresponding set of projects are WBB. We have shown that a set of projects is feasible if and only if the corresponding set of items is feasible and equated the objective functions. Thus, since $\operatorname{KNAPSACK}$ admits an FPTAS, we have that UWO admits an FPTAS in the additive valuations setting. 
\end{proof}

\section{Proof of Lemma~\ref{lem:maxpe1}}
\label{app:maxpe1}
\lemmaxpeone*
\begin{proof}
    We give a reduction from the NP-complete Exact Cover by 3-Sets (X3C) problem: given a set $Z=\{z_1, \ldots, z_{3q}\}$ (with cardinality a multiple of $3$) and a family $S$ of 3-element subsets (triples) of $Z$, the problem asks if there is a subfamily $S^{\prime}\subseteq S$ such that every element in $Z$ is contained in exactly one triple of $S^{\prime}$. Given any instance of X3C, we construct an instance of $\maxpe\geq1$ such that we have:
    \begin{itemize}
        \item an agent $i\in N$ for each $z_i\in Z$, each with identical budget of $b_i = \frac{1}{3} + \frac{1}{n}$
        \item a project $j\in M$ for each triple $s_j\in S$
        \item agent valuations $v_{ij} = \frac{1}{3} + \frac{1}{n}$ if $z_i\in s_j$ and $0$ otherwise.
    \end{itemize}
    The $\maxpe\geq1$ decision problem becomes
    \begin{align*}
        \max_{W\subseteq M} \ \  \sum_{i\in N} \left(\frac{1}{3} + \frac{1}{n}\right) \min(1, \lvert \{j\in W: z_i\in s_j\}\rvert)  - \lvert W \rvert \geq 1
    \end{align*}
    Note that an agent $i\in N$ contributes $0$ to the above summation if a bundle does not contain any project $j$ for which $z_i\in s_j$ and they contribute exactly $\frac{1}{3} + \frac{1}{n}$ otherwise. Thus, the maximum this summation can evaluate to is $n\cdot (\frac{1}{3} + \frac{1}{n}) = \frac{n}{3} + 1$.
    \par If there exists a subfamily $S^{\prime}\subseteq S$ such that every element in $Z$ is contained in exactly one triple of $S^{\prime}$, then the corresponding $\maxpe\geq1$ instance is a $\mathtt{Yes}$ instance. To see this, consider $W=\{j: s_j\in S^{\prime}\}$. Note that $\lvert W \rvert = \frac{n}{3}$ since this is the number of triples required to cover a set of size $n$. By the transformation, each agent attributes non-zero value to exactly one project in $W$. Thus, the $\maxpe$ expression becomes $n\cdot (\frac{1}{3} + \frac{1}{n}) - \frac{n}{3} = 1$. 
    \par In the other direction, if the $\maxpe\geq 1$ instance is a $\mathtt{Yes}$ instance, let $W$ be a bundle for which the excess payment extraction is weakly greater than 1. First note that $\sum_{i\in N} b_i = \frac{n}{3}+1$ so $\lvert W\rvert \leq \frac{n}{3}+1$ by WBB. However, if $\lvert W\rvert = \frac{n}{3}+1$, then the $\maxpe$ expression is upper bounded by $\frac{n}{3}+1 - (\frac{n}{3}+1) = 0$. If, on the other hand, $\lvert W\rvert < \frac{n}{3}$, because each project is valued by exactly 3 agents, the $\maxpe$ expression is upper bounded by $3\cdot \lvert W \rvert (\frac{1}{3} + \frac{1}{n}) - \lvert W \rvert = \frac{3\cdot \lvert W \rvert}{n} < \frac{3(n/3)}{n} = 1$. Thus, $\lvert W\rvert = \frac{n}{3}$. 
    \par Consider the subfamily $S^{\prime} = \{s_j: j\in W\}$. We wish to show $S^{\prime}$ constitutes an exact cover of $Z$. If it does not, it either does not cover $Z$ or it contains two subsets which are not disjoint. In the former case, there is an element $z\in Z$ which is not contained in any triple in $S^{\prime}$ which means there is an agent in the $\maxpe\geq1$ instance whose total valuation is $0$. This means the $\maxpe$ expression is upper bounded by $(n-1)(\frac{1}{3} + \frac{1}{n}) - \frac{n}{3} = 2/3 - 1/n < 1$, which is a contradiction. In the latter case, there exist $s_j, s_k \in S^{\prime}$ such that $s_j\cap s_k \neq \emptyset$, which means there exists an agent in the corresponding $\maxpe\geq1$ instance with non-zero valuations for both $j$ and $k$. Because there are exactly $3$ agents with non-zero valuations for each of the $n/3$ projects in $W$, this means there must be some agent who has zero valuation for every project in $W$ and we reach a contradiction. Therefore, $S^{\prime}$ constitutes an exact cover of $Z$ and the X3C instance is a $\mathtt{Yes}$ instance.
\end{proof}

\section{Proof of Lemma~\ref{lem:mpe_sm_general}}
\lemmaxpesm*
\begin{proof}
    Consider the following ILP: \\
    \begin{equation*}
    \begin{array}{ll@{}ll}
        \text{maximize}  & \displaystyle\sum_{i\in N} \min(b_i, z_i)\cdot x_i  - \displaystyle\sum_{j\in M} C_j\cdot y_j &\\
        \text{subject to}& x_i \leq y_j,    & \forall i\in N, \forall j\in D_i \\
                         & x_i \in \{0,1\}, & \forall i\in N\\
                         & y_j \in \{0,1\}, & \forall j\in M\\
    \end{array}
    \end{equation*}
    We claim that the above ILP describes our $\maxpe$ problem. For any feasible solution to $\maxpe$, $W$, let $x_i=1$ iff $D_i\subseteq W$ and $y_j=1$ iff $j\in W$. It is apparent that the objective function of the ILP is equivalent to that of $\maxpe$. Furthermore, for each $i\in N$ with $x_i=1$, we have $y_j=1 \hspace{0.2cm} \forall j \in D_i \subseteq W$. Thus, the constraints are satisfied and we have a corresponding feasible solution to our ILP.
    \\
    Notice that $\min(b_i, z_i)$ is a constant for each agent $i$, and our ILP is thus identical to that given by \citet{BCG+15a}. As shown by \citet{BCG+15a}, the LP relaxation of this ILP yields a totally unimodular constraint matrix and thus the $\maxpe$ problem is polynomial time solvable. 
\end{proof}

\section{Set-Union Knapsack Problem FPTAS}
\label{app:sukp}
The \textit{set-union knapsack problem} (SUKP) is a well-studied generalization of the classical knapsack problem, wherein items have values and are composed of elements, which may occur in multiple items and are associated with costs \citep{GNY94a,Arul14a,WeHa19a,HXW+18a}. The SUKP solution maximizes the total value of the items in the knapsack while ensuring the weight of the elements in the items' union does not exceed the given knapsack capacity. 
In this section, we show that under the laminar item set family restriction, the problem admits an FPTAS. To do so, we provide an approximation-preserving reduction to \textit{knapsack on chordal conflict graphs} (defined in Section \ref{sec:laminar_sm}).

SUKP takes as input a set of $n$ items and a universe of elements $U$ where each item $i$ corresponds to an \textit{element set}, $L_i$, such that $\bigcup_{i=1}^{n} L_i = U$. Each item $i$ has a non-negative value denoted $v_i$, and each element $j$ has a non-negative weight, $\omega_j$. The knapsack has a capacity of $K$. 

We denote the item set family as $\mathbf{L} = \{L_1, \ldots, L_n\}$, the item values as $\mathbf{v}$, the element weights as $\boldsymbol\omega$, and denote an instance of SUKP as $I=\langle \mathbf{L}, \mathbf{v}, \boldsymbol\omega, K\rangle$. We henceforth use the notation $[n]=\{1,\dots, n\}$. The SUKP problem is to find a subset of items such that the total value is maximized and the weight of the elements in the union of the items does not exceed the capacity. Formally, the solution can be written

\begin{equation}
\begin{aligned}
    \max_{T\subseteq [n]} \quad & \displaystyle\sum_{i\in T} v_i \hspace{0.2cm} \nonumber\\
    \textrm{subject to} \quad & \displaystyle\sum_{j\in \bigcup\limits_{i\in T} L_i} \omega_j \leq K \nonumber \\
\end{aligned}
\end{equation}

\begin{definition}[Laminar SUKP]
    We call an instance of SUKP \textit{laminar} if the items' element sets constitute a laminar set family, i.e. if for every $L_i, L_j \in \mathbf{L}$, the intersection of $L_i$ and $L_j$ is either empty, or equals $L_i$, or equals $L_j$.
\end{definition}

Next, we prove two lemmas which will help us to connect Laminar SUKP with KCG.

\begin{lemma}\label{lem:chordal_sukp}
    Given an instance $I=\langle \mathbf{L}, \mathbf{v}, \boldsymbol\omega, K\rangle$ of laminar SUKP, the undirected graph $G=(V,E)$ where $V=[n]$ and $E=\{(i,j) | L_i\subseteq L_j\}$ is \textit{chordal}.
\end{lemma}
\begin{proof}
    Consider such a graph $G=(V,E)$. Let $E^c\subseteq E$ be a cycle of length four or more. Pick any edge $(i,j)\in E^c$ and assume without loss of generality that $L_i\subseteq L_j$. Now pick any edge $(i,k)\in E^c$, which must exist because $E^c$ is a cycle. We have that $L_i\cap L_k\neq \emptyset$ by the construction of the graph and thus $L_k\cap L_j \neq \emptyset$. Since $\mathbf{L}$ is a laminar set family, this means that $L_j$ and $L_k$ form a subset relation, which implies $(j,k)\in E$. $(j,k)$ is a chord since it connects two vertices in $E^c$ and cannot be in $E^c$ since it would form a triangle, contradicting $E^c$ being a cycle of length four or more. 
\end{proof}

\begin{lemma}\label{lem:sukp_full_soln}
    Given an instance $I=\langle \mathbf{L}, \mathbf{v}, \boldsymbol\omega, K\rangle$ of laminar SUKP and any set $T\subseteq [n]$, let $T^\prime = \{i\in [n] \vert L_i\subseteq \bigcup_{j\in T} L_j\}$. $T^\prime$ is feasible.
\end{lemma}
\begin{proof}
    Consider $i\in[n]$ such that $i\not\in T$ and $L_i\subseteq \bigcup_{j\in T} L_j$. It suffices to show that $\{T\cup i\}$ is a feasible solution. This is true because $e\in L_i \implies e\in \bigcup_{j\in T} L_j$ so the weights of $i$'s element set are already included in the knapsack constraint.
\end{proof} 

The following theorem makes a novel contribution to the literature around the set-union knapsack problem by identifying a natural restriction which admits an FPTAS. 

\begin{theorem}\label{thm:lam_sukp}
    Laminar SUKP admits an FPTAS.
\end{theorem}

\begin{proof}
    Let $I=\langle \mathbf{L}, \mathbf{v}, \boldsymbol\omega, K\rangle$ be an instance of laminar SUKP. We proceed by an approximation-preserving reduction to an instance of KCG on chordal graphs as follows:
    \begin{itemize}
        \item $V = [n]$
        \item $p_i = \displaystyle\sum_{\{j\in [n] | L_j\subseteq L_i\}} v_j$
        \item $w_i = \displaystyle\sum_{e\in L_i} \omega_e$
        \item $B = K$
        \item $(i,j)\in E$ for each distinct $i,j\in [n]$ with $L_i\subseteq L_j$
    \end{itemize}

    Let $S$ be a feasible solution to the instance of KCG reduced from $I$, as given above. We posit that $T^\prime = \{i\in [n] \vert L_i\subseteq \bigcup_{j\in S} L_j \}$ is a feasible solution to $I$. Since $S$ is a feasible solution to KCG, we know
    $$ \sum_{i\in S} w_i \leq B \iff \sum_{i\in S} \sum_{e\in L_i} \omega_e \leq K $$
    Since $\mathbf{L}$ is a laminar set family, we know that the sets are either disjoint or form a subset relation. However, since S is feasible, we know that no two members of $\{L_i\}_{i\in S}$ can form a subset relation and thus all members are disjoint. Thus, we can rewrite the above as
    $$\sum_{e\in \bigcup\limits_{i\in S} L_i} \omega_e \leq K $$
    and we have that $S$ is a feasible solution to $I$. From Lemma \ref{lem:sukp_full_soln}, we can thus conclude that $T^\prime$ is a feasible solution to the SUKP instance $I$.

    We now show that the objective value of the KCG solution, $S$, is equal to that of the resulting SUKP solution $T^\prime$:
    \begin{align}
    \sum_{i\in S} p_i & = \sum_{i\in S} \sum_{\{j\in [n] \vert L_j\subseteq L_i\}} v_j \nonumber \\
    & = \sum_{\{j\in [n]\vert L_j\subseteq \bigcup\limits_{i\in S} L_i \}} v_j \nonumber \\
    & = \sum_{j\in T^\prime} v_j \nonumber
    \end{align}

    It is apparent that the reduction given above runs in polynomial time. Furthermore, by Lemma \ref{lem:chordal_sukp} we have that $G$ is a chordal graph and thus the reduced instance admits an FPTAS \citep{PfSc09a}. Since our reduction equated the objective values of each problem, it is approximation-preserving and we therefore have that Laminar SUKP admits an FPTAS.
\end{proof}

%% file: main.bbl

\begin{thebibliography}{45}


\ifx \showCODEN    \undefined \def \showCODEN     #1{\unskip}     \fi
\ifx \showDOI      \undefined \def \showDOI       #1{#1}\fi
\ifx \showISBNx    \undefined \def \showISBNx     #1{\unskip}     \fi
\ifx \showISBNxiii \undefined \def \showISBNxiii  #1{\unskip}     \fi
\ifx \showISSN     \undefined \def \showISSN      #1{\unskip}     \fi
\ifx \showLCCN     \undefined \def \showLCCN      #1{\unskip}     \fi
\ifx \shownote     \undefined \def \shownote      #1{#1}          \fi
\ifx \showarticletitle \undefined \def \showarticletitle #1{#1}   \fi
\ifx \showURL      \undefined \def \showURL       {\relax}        \fi
\providecommand\bibfield[2]{#2}
\providecommand\bibinfo[2]{#2}
\providecommand\natexlab[1]{#1}
\providecommand\showeprint[2][]{arXiv:#2}

\bibitem[Archer et~al\mbox{.}(2004)]%
        {APT+04a}
\bibfield{author}{\bibinfo{person}{Aaron Archer}, \bibinfo{person}{Christos
  Papadimitriou}, \bibinfo{person}{Kunal Talwar}, {and}
  \bibinfo{person}{{\'E}va Tardos}.} \bibinfo{year}{2004}\natexlab{}.
\newblock \showarticletitle{An Approximate Truthful Mechanism for Combinatorial
  Auctions with Single Parameter Agents}.
\newblock \bibinfo{journal}{\emph{Internet Mathematics}} \bibinfo{volume}{1},
  \bibinfo{number}{2} (\bibinfo{year}{2004}), \bibinfo{pages}{129--150}.
\newblock


\bibitem[Arulselvan(2014)]%
        {Arul14a}
\bibfield{author}{\bibinfo{person}{Ashwin Arulselvan}.}
  \bibinfo{year}{2014}\natexlab{}.
\newblock \showarticletitle{A Note on the Set Union Knapsack Problem}.
\newblock \bibinfo{journal}{\emph{Discrete Applied Mathematics}}
  \bibinfo{volume}{169} (\bibinfo{year}{2014}), \bibinfo{pages}{214--218}.
\newblock


\bibitem[Aziz(2020)]%
        {Aziz20b}
\bibfield{author}{\bibinfo{person}{H. Aziz}.} \bibinfo{year}{2020}\natexlab{}.
\newblock \showarticletitle{Strategyproof multi-item exchange under
  single-minded dichotomous preferences}.
\newblock \bibinfo{journal}{\emph{Autonomous Agents and Multi-Agent Systems}}
  \bibinfo{volume}{34}, \bibinfo{number}{1} (\bibinfo{year}{2020}).
\newblock


\bibitem[Aziz and Ganguly(2021)]%
        {AzGa21a}
\bibfield{author}{\bibinfo{person}{H. Aziz} {and} \bibinfo{person}{A.
  Ganguly}.} \bibinfo{year}{2021}\natexlab{}.
\newblock \showarticletitle{Participatory Funding Coordination: Model, Axioms
  and Rules}. In \bibinfo{booktitle}{\emph{International Conference on
  Algorithmic Decision Theory}} \emph{(\bibinfo{series}{Lecture Notes in
  Computer Science}, Vol.~\bibinfo{volume}{13023})}.
  \bibinfo{publisher}{Springer}, \bibinfo{pages}{409--423}.
\newblock


\bibitem[Aziz et~al\mbox{.}(2018)]%
        {ALT18a}
\bibfield{author}{\bibinfo{person}{H. Aziz}, \bibinfo{person}{B.~E. Lee}, {and}
  \bibinfo{person}{N. Talmon}.} \bibinfo{year}{2018}\natexlab{}.
\newblock \showarticletitle{Proportionally Representative Participatory
  Budgeting: Axioms and Algorithms}. In \bibinfo{booktitle}{\emph{Proceedings
  of the 17th International Conference on Autonomous Agents and MultiAgent
  Systems, {AAMAS} 2018, Stockholm, Sweden, July 10-15, 2018}}.
  \bibinfo{pages}{23--31}.
\newblock


\bibitem[Aziz and Shah(2020)]%
        {AzSh20a}
\bibfield{author}{\bibinfo{person}{H. Aziz} {and} \bibinfo{person}{N. Shah}.}
  \bibinfo{year}{2020}\natexlab{}.
\newblock \showarticletitle{Participatory Budgeting: Models and Approaches}.
\newblock In \bibinfo{booktitle}{\emph{Pathways between Social Science and
  Computational Social Science: Theories, Methods and Interpretations}},
  \bibfield{editor}{\bibinfo{person}{T.~Rudas} {and}
  \bibinfo{person}{P.~G\'{a}bor}} (Eds.). \bibinfo{publisher}{Springer}.
\newblock


\bibitem[Benade et~al\mbox{.}(2017)]%
        {BNPS17a}
\bibfield{author}{\bibinfo{person}{G. Benade}, \bibinfo{person}{S. Nath},
  \bibinfo{person}{A.~D. Procaccia}, {and} \bibinfo{person}{N. Shah}.}
  \bibinfo{year}{2017}\natexlab{}.
\newblock \showarticletitle{Preference Elicitation For Participatory
  Budgeting}. In \bibinfo{booktitle}{\emph{31st}}. \bibinfo{publisher}{AAAI
  Press}, \bibinfo{pages}{376--382}.
\newblock


\bibitem[Birmpas et~al\mbox{.}(2015)]%
        {BCG+15a}
\bibfield{author}{\bibinfo{person}{Georgios Birmpas}, \bibinfo{person}{Costas
  Courcoubetis}, \bibinfo{person}{Ioannis Giotis}, {and}
  \bibinfo{person}{Evangelos Markakis}.} \bibinfo{year}{2015}\natexlab{}.
\newblock \showarticletitle{Cost-{{Sharing Models}} in {{Participatory
  Sensing}}}.
\newblock In \bibinfo{booktitle}{\emph{Algorithmic {{Game Theory}}}},
  \bibfield{editor}{\bibinfo{person}{Martin Hoefer}} (Ed.).
  Vol.~\bibinfo{volume}{9347}. \bibinfo{publisher}{{Springer}},
  \bibinfo{pages}{43--56}.
\newblock


\bibitem[Birmpas et~al\mbox{.}(2019)]%
        {birmpas2019cost}
\bibfield{author}{\bibinfo{person}{G. Birmpas}, \bibinfo{person}{E. Markakis},
  {and} \bibinfo{person}{G. Sch{\"a}fer}.} \bibinfo{year}{2019}\natexlab{}.
\newblock \showarticletitle{Cost sharing over combinatorial domains:
  Complement-free cost functions and beyond}. In \bibinfo{booktitle}{\emph{27th
  Annual European Symposium on Algorithms, ESA 2019}}. Schloss
  Dagstuhl-Leibniz-Zentrum fur Informatik GmbH, Dagstuhl Publishing,
  \bibinfo{pages}{1--17}.
\newblock


\bibitem[Brandl et~al\mbox{.}(2022)]%
        {BBGPSS22}
\bibfield{author}{\bibinfo{person}{F. Brandl}, \bibinfo{person}{F. Brandt},
  \bibinfo{person}{M. Greger}, \bibinfo{person}{D. Peter}, \bibinfo{person}{C.
  Stricker}, {and} \bibinfo{person}{W. Suksompong}.}
  \bibinfo{year}{2022}\natexlab{}.
\newblock \showarticletitle{Funding public projects: A case for the Nash
  product rule}.
\newblock \bibinfo{journal}{\emph{Journal of Mathematical Economics}}
  \bibinfo{volume}{99} (\bibinfo{year}{2022}).
\newblock


\bibitem[Br{\^a}nzei et~al\mbox{.}(2016)]%
        {BLM16a}
\bibfield{author}{\bibinfo{person}{Simina Br{\^a}nzei},
  \bibinfo{person}{Yuezhou Lv}, {and} \bibinfo{person}{Ruta Mehta}.}
  \bibinfo{year}{2016}\natexlab{}.
\newblock \showarticletitle{To Give or Not to Give: Fair Division for Single
  Minded Valuations}. In \bibinfo{booktitle}{\emph{Proceedings of the
  {{Twenty-Fifth International Joint Conference}} on {{Artificial
  Intelligence}}}} \emph{(\bibinfo{series}{{{IJCAI}}'16})}.
  \bibinfo{publisher}{{AAAI Press}}, \bibinfo{address}{{New York, New York,
  USA}}, \bibinfo{pages}{123--129}.
\newblock
\showISBNx{978-1-57735-770-4}


\bibitem[Buterin et~al\mbox{.}(2019)]%
        {BHW19a}
\bibfield{author}{\bibinfo{person}{V. Buterin}, \bibinfo{person}{Z. Hitzig},
  {and} \bibinfo{person}{E.~G. Weyl}.} \bibinfo{year}{2019}\natexlab{}.
\newblock \showarticletitle{{A Flexible Design for Funding Public Goods}}.
\newblock \bibinfo{journal}{\emph{Management Science}} \bibinfo{volume}{65},
  \bibinfo{number}{11} (\bibinfo{year}{2019}), \bibinfo{pages}{5171--5187}.
\newblock


\bibitem[Chen et~al\mbox{.}(2022)]%
        {chen2022participatory}
\bibfield{author}{\bibinfo{person}{Jiehua Chen}, \bibinfo{person}{Martin
  Lackner}, {and} \bibinfo{person}{Jan Maly}.} \bibinfo{year}{2022}\natexlab{}.
\newblock \showarticletitle{Participatory Budgeting with Donations and
  Diversity Constraints}. In \bibinfo{booktitle}{\emph{Proceedings of the AAAI
  Conference on Artificial Intelligence}}, Vol.~\bibinfo{volume}{36}.
  \bibinfo{pages}{9323--9330}.
\newblock


\bibitem[Chen et~al\mbox{.}(2004)]%
        {CDS04a}
\bibfield{author}{\bibinfo{person}{Ning Chen}, \bibinfo{person}{Xiaotie Deng},
  {and} \bibinfo{person}{Xiaoming Sun}.} \bibinfo{year}{2004}\natexlab{}.
\newblock \showarticletitle{On Complexity of Single-Minded Auction}.
\newblock \bibinfo{journal}{\emph{J. Comput. System Sci.}}
  \bibinfo{volume}{69}, \bibinfo{number}{4} (\bibinfo{year}{2004}),
  \bibinfo{pages}{675--687}.
\newblock


\bibitem[Damle et~al\mbox{.}(2019)]%
        {damle2019civic}
\bibfield{author}{\bibinfo{person}{S. Damle}, \bibinfo{person}{M.~H. Moti},
  \bibinfo{person}{P. Chandra}, {and} \bibinfo{person}{S. Gujar}.}
  \bibinfo{year}{2019}\natexlab{}.
\newblock \showarticletitle{Civic Crowdfunding for Agents with Negative
  Valuations and Agents with Asymmetric Beliefs}. In
  \bibinfo{booktitle}{\emph{Proceedings of the 28th International Joint
  Conference on Artificial Intelligence}}. \bibinfo{pages}{208–214}.
\newblock
\showISBNx{9780999241141}


\bibitem[Devanur et~al\mbox{.}(2020)]%
        {DGS+20a}
\bibfield{author}{\bibinfo{person}{Nikhil~R. Devanur}, \bibinfo{person}{Kira
  Goldner}, \bibinfo{person}{Raghuvansh~R. Saxena}, \bibinfo{person}{Ariel
  Schvartzman}, {and} \bibinfo{person}{S.~Matthew Weinberg}.}
  \bibinfo{year}{2020}\natexlab{}.
\newblock \showarticletitle{Optimal Mechanism Design for Single-Minded Agents}.
  In \bibinfo{booktitle}{\emph{Proceedings of the 21st {{ACM Conference}} on
  {{Economics}} and {{Computation}}}}. \bibinfo{pages}{193--256}.
\newblock


\bibitem[Dobzinski et~al\mbox{.}(2018)]%
        {tim2018}
\bibfield{author}{\bibinfo{person}{S. Dobzinski}, \bibinfo{person}{A. Mehta},
  \bibinfo{person}{T. Roughgarden}, {and} \bibinfo{person}{M. Sundararajan}.}
  \bibinfo{year}{2018}\natexlab{}.
\newblock \showarticletitle{Is Shapley cost sharing optimal?}
\newblock \bibinfo{journal}{\emph{Games and Economic Behavior}}
  \bibinfo{volume}{108} (\bibinfo{year}{2018}), \bibinfo{pages}{130--138}.
\newblock
\showISSN{0899-8256}
\newblock
\shownote{Special Issue in Honor of Lloyd Shapley: Seven Topics in Game
  Theory}.


\bibitem[Dobzinski and Ovadia(2017)]%
        {shahar17}
\bibfield{author}{\bibinfo{person}{S. Dobzinski} {and} \bibinfo{person}{S.
  Ovadia}.} \bibinfo{year}{2017}\natexlab{}.
\newblock \showarticletitle{Combinatorial cost sharing}. In
  \bibinfo{booktitle}{\emph{Proceedings of the 2017 ACM Conference on Economics
  and Computation}}. \bibinfo{pages}{387--404}.
\newblock


\bibitem[Elkind et~al\mbox{.}(2017)]%
        {EFSS17a}
\bibfield{author}{\bibinfo{person}{E. Elkind}, \bibinfo{person}{P.
  Faliszewski}, \bibinfo{person}{P. Skowron}, {and} \bibinfo{person}{A.
  Slinko}.} \bibinfo{year}{2017}\natexlab{}.
\newblock \showarticletitle{Properties of Multiwinner Voting Rules}.
\newblock \bibinfo{journal}{\emph{Social Choice and Welfare}}
  (\bibinfo{year}{2017}).
\newblock


\bibitem[Garey and Johnson(1979)]%
        {GaJo79a}
\bibfield{author}{\bibinfo{person}{M.~R. Garey} {and} \bibinfo{person}{D.~S.
  Johnson}.} \bibinfo{year}{1979}\natexlab{}.
\newblock \bibinfo{booktitle}{\emph{Computers and Intractability: A Guide to
  the Theory of NP-Completeness}}.
\newblock \bibinfo{publisher}{W. H. Freeman}.
\newblock


\bibitem[Goldschmidt et~al\mbox{.}(1994)]%
        {GNY94a}
\bibfield{author}{\bibinfo{person}{Olivier Goldschmidt}, \bibinfo{person}{David
  Nehme}, {and} \bibinfo{person}{Gang Yu}.} \bibinfo{year}{1994}\natexlab{}.
\newblock \showarticletitle{Note: {{On}} the Set-Union Knapsack Problem}.
\newblock \bibinfo{journal}{\emph{Naval Research Logistics (NRL)}}
  \bibinfo{volume}{41}, \bibinfo{number}{6} (\bibinfo{year}{1994}),
  \bibinfo{pages}{833--842}.
\newblock


\bibitem[Green and Laffont(1979)]%
        {GrLa79a}
\bibfield{author}{\bibinfo{person}{J.~R. Green} {and} \bibinfo{person}{J-J.
  Laffont}.} \bibinfo{year}{1979}\natexlab{}.
\newblock \bibinfo{booktitle}{\emph{Incentives in Public Decision-Making}}.
  Vol.~\bibinfo{volume}{1}.
\newblock \bibinfo{publisher}{North-Holland}.
\newblock


\bibitem[Groves(1973)]%
        {Grov73a}
\bibfield{author}{\bibinfo{person}{T. Groves}.}
  \bibinfo{year}{1973}\natexlab{}.
\newblock \showarticletitle{Incentives in Teams}.
\newblock \bibinfo{journal}{\emph{Econometrica}}  \bibinfo{volume}{41}
  (\bibinfo{year}{1973}), \bibinfo{pages}{617--631}.
\newblock


\bibitem[He et~al\mbox{.}(2018)]%
        {HXW+18a}
\bibfield{author}{\bibinfo{person}{Yichao He}, \bibinfo{person}{Haoran Xie},
  \bibinfo{person}{Tak-Lam Wong}, {and} \bibinfo{person}{Xizhao Wang}.}
  \bibinfo{year}{2018}\natexlab{}.
\newblock \showarticletitle{A Novel Binary Artificial Bee Colony Algorithm for
  the Set-Union Knapsack Problem}.
\newblock \bibinfo{journal}{\emph{Future Generation Computer Systems}}
  \bibinfo{volume}{78} (\bibinfo{year}{2018}), \bibinfo{pages}{77--86}.
\newblock


\bibitem[Hershkowitz et~al\mbox{.}(2021)]%
        {Hersh2021}
\bibfield{author}{\bibinfo{person}{D.~Ellis Hershkowitz},
  \bibinfo{person}{Anson Kahng}, \bibinfo{person}{Dominik Peters}, {and}
  \bibinfo{person}{Ariel~D. Procaccia}.} \bibinfo{year}{2021}\natexlab{}.
\newblock \showarticletitle{District-Fair Participatory Budgeting}. In
  \bibinfo{booktitle}{\emph{Proceedings of the AAAI Conference on Artificial
  Intelligence}}, Vol.~\bibinfo{volume}{35}. \bibinfo{pages}{5464--5471}.
\newblock


\bibitem[Hoefer and Kesselheim(2012)]%
        {MhKt15}
\bibfield{author}{\bibinfo{person}{Martin Hoefer} {and} \bibinfo{person}{Thomas
  Kesselheim}.} \bibinfo{year}{2012}\natexlab{}.
\newblock \showarticletitle{Secondary spectrum auctions for symmetric and
  submodular bidders}. In \bibinfo{booktitle}{\emph{Proceedings of the 13th ACM
  Conference on Electronic Commerce}}. \bibinfo{pages}{657--671}.
\newblock


\bibitem[Jain et~al\mbox{.}(2020)]%
        {PaTa20}
\bibfield{author}{\bibinfo{person}{P. Jain}, \bibinfo{person}{K. Sornat}, {and}
  \bibinfo{person}{N. Talmon}.} \bibinfo{year}{2020}\natexlab{}.
\newblock \showarticletitle{Participatory Budgeting with Project Interactions}.
  In \bibinfo{booktitle}{\emph{Proceedings of the Twenty-Ninth International
  Conference on International Joint Conferences on Artificial Intelligence}}.
  \bibinfo{pages}{386--392}.
\newblock


\bibitem[Kilgour(2016)]%
        {Kilg16a}
\bibfield{author}{\bibinfo{person}{D.~M. Kilgour}.}
  \bibinfo{year}{2016}\natexlab{}.
\newblock \showarticletitle{Approval elections with a variable number of
  winners}.
\newblock \bibinfo{journal}{\emph{Theory and Decision}} \bibinfo{volume}{81},
  \bibinfo{number}{2} (\bibinfo{year}{2016}), \bibinfo{pages}{199--211}.
\newblock


\bibitem[Lawler(1979)]%
        {fptas}
\bibfield{author}{\bibinfo{person}{E.~L. Lawler}.}
  \bibinfo{year}{1979}\natexlab{}.
\newblock \showarticletitle{Fast approximation algorithms for knapsack
  problems}.
\newblock \bibinfo{journal}{\emph{Mathematics of Operations Research}}
  \bibinfo{volume}{4}, \bibinfo{number}{4} (\bibinfo{year}{1979}),
  \bibinfo{pages}{339--356}.
\newblock


\bibitem[Lu and Boutilier(2011)]%
        {lu2011budgeted}
\bibfield{author}{\bibinfo{person}{T. Lu} {and} \bibinfo{person}{C.
  Boutilier}.} \bibinfo{year}{2011}\natexlab{}.
\newblock \showarticletitle{Budgeted social choice: From consensus to
  personalized decision making}. In \bibinfo{booktitle}{\emph{Twenty-Second
  International Joint Conference on Artificial Intelligence}}.
  \bibinfo{pages}{280--286}.
\newblock


\bibitem[Moulin(1994)]%
        {moulin1994serial}
\bibfield{author}{\bibinfo{person}{H. Moulin}.}
  \bibinfo{year}{1994}\natexlab{}.
\newblock \showarticletitle{Serial cost-sharing of excludable public goods}.
\newblock \bibinfo{journal}{\emph{The Review of Economic Studies}}
  \bibinfo{volume}{61}, \bibinfo{number}{2} (\bibinfo{year}{1994}),
  \bibinfo{pages}{305--325}.
\newblock


\bibitem[{Nehme-Haily}(1995)]%
        {Nehm95a}
\bibfield{author}{\bibinfo{person}{David~Antoine {Nehme-Haily}}.}
  \bibinfo{year}{1995}\natexlab{}.
\newblock \emph{\bibinfo{title}{The Set-Union Knapsack Problem}}.
\newblock \bibinfo{thesistype}{Ph.\,D. Dissertation}. \bibinfo{school}{The
  University of Texas at Austin}.
\newblock
\showISBNx{9798209176398}


\bibitem[Ohseto(2000)]%
        {ohseto2000}
\bibfield{author}{\bibinfo{person}{S. Ohseto}.}
  \bibinfo{year}{2000}\natexlab{}.
\newblock \showarticletitle{Characterizations of strategy-proof mechanisms for
  excludable versus nonexcludable public projects}.
\newblock \bibinfo{journal}{\emph{Games and Economic Behavior}}
  \bibinfo{volume}{32}, \bibinfo{number}{1} (\bibinfo{year}{2000}),
  \bibinfo{pages}{51--66}.
\newblock


\bibitem[Pferschy and Schauer(2009)]%
        {PfSc09a}
\bibfield{author}{\bibinfo{person}{Ulrich Pferschy} {and}
  \bibinfo{person}{Joachim Schauer}.} \bibinfo{year}{2009}\natexlab{}.
\newblock \showarticletitle{The {{Knapsack Problem}} with {{Conflict Graphs}}.}
\newblock \bibinfo{journal}{\emph{J. Graph Algorithms Appl.}}
  \bibinfo{volume}{13}, \bibinfo{number}{2} (\bibinfo{year}{2009}),
  \bibinfo{pages}{233--249}.
\newblock


\bibitem[Samuelson(1954)]%
        {samuelson1954pure}
\bibfield{author}{\bibinfo{person}{Paul~A Samuelson}.}
  \bibinfo{year}{1954}\natexlab{}.
\newblock \showarticletitle{The pure theory of public expenditure}.
\newblock \bibinfo{journal}{\emph{The review of economics and statistics}}
  (\bibinfo{year}{1954}), \bibinfo{pages}{387--389}.
\newblock


\bibitem[Shah(2007)]%
        {shah2007}
\bibfield{author}{\bibinfo{person}{A. Shah}.} \bibinfo{year}{2007}\natexlab{}.
\newblock \bibinfo{booktitle}{\emph{Participatory budgeting.}}
\newblock \bibinfo{publisher}{World Bank Publications}.
\newblock


\bibitem[Stolicki et~al\mbox{.}(2020)]%
        {SST20a}
\bibfield{author}{\bibinfo{person}{Dariusz Stolicki},
  \bibinfo{person}{Stanis{\l}aw Szufa}, {and} \bibinfo{person}{Nimrod Talmon}.}
  \bibinfo{year}{2020}\natexlab{}.
\newblock \showarticletitle{Pabulib: {{A Participatory Budgeting Library}}}.
\newblock \bibinfo{journal}{\emph{arXiv:2012.06539 [cs]}} (\bibinfo{date}{Dec.}
  \bibinfo{year}{2020}).
\newblock


\bibitem[Talmon and Faliszewski(2019)]%
        {NtPf19}
\bibfield{author}{\bibinfo{person}{N. Talmon} {and} \bibinfo{person}{P.
  Faliszewski}.} \bibinfo{year}{2019}\natexlab{}.
\newblock \showarticletitle{A framework for approval-based budgeting methods}.
  In \bibinfo{booktitle}{\emph{Proceedings of the AAAI Conference on Artificial
  Intelligence}}, Vol.~\bibinfo{volume}{33}. \bibinfo{pages}{2181--2188}.
\newblock


\bibitem[Vazirani(2001)]%
        {Vazi01a}
\bibfield{author}{\bibinfo{person}{V.~V. Vazirani}.}
  \bibinfo{year}{2001}\natexlab{}.
\newblock \bibinfo{booktitle}{\emph{Approximation Algorithms}}.
\newblock \bibinfo{publisher}{Springer}.
\newblock


\bibitem[Vries et~al\mbox{.}(2020)]%
        {VNS20a}
\bibfield{editor}{\bibinfo{person}{M.~S.~De Vries}, \bibinfo{person}{J. Nemec},
  {and} \bibinfo{person}{D. {\v S}pa{\v c}ek}} (Eds.).
  \bibinfo{year}{2020}\natexlab{}.
\newblock \bibinfo{booktitle}{\emph{International Trends in Participatory
  Budgeting: Between Trivial Pursuits and Best Practices}}.
\newblock \bibinfo{publisher}{Springer}.
\newblock


\bibitem[Wagner and Meir(2021)]%
        {wagnervcg}
\bibfield{author}{\bibinfo{person}{J. Wagner} {and} \bibinfo{person}{R. Meir}.}
  \bibinfo{year}{2021}\natexlab{}.
\newblock \showarticletitle{A VCG Adaptation for Participatory Budgeting}.
\newblock  (\bibinfo{year}{2021}).
\newblock


\bibitem[Wang et~al\mbox{.}(2021)]%
        {guo}
\bibfield{author}{\bibinfo{person}{G. Wang}, \bibinfo{person}{R.Guo},
  \bibinfo{person}{Y. Sakurai}, \bibinfo{person}{A. Babar}, {and}
  \bibinfo{person}{M. Guo}.} \bibinfo{year}{2021}\natexlab{}.
\newblock \showarticletitle{Mechanism Design for Public Projects via Neural
  Networks}. In \bibinfo{booktitle}{\emph{Proceedings of the 20th International
  Conference on Autonomous Agents and Multi-Agent Systems}}.
  \bibinfo{pages}{1380--1388}.
\newblock


\bibitem[Wei and Hao(2019)]%
        {WeHa19a}
\bibfield{author}{\bibinfo{person}{Zequn Wei} {and} \bibinfo{person}{Jin-Kao
  Hao}.} \bibinfo{year}{2019}\natexlab{}.
\newblock \showarticletitle{Iterated Two-Phase Local Search for the {{Set-Union
  Knapsack Problem}}}.
\newblock \bibinfo{journal}{\emph{Future Generation Computer Systems}}
  \bibinfo{volume}{101} (\bibinfo{date}{Dec.} \bibinfo{year}{2019}),
  \bibinfo{pages}{1005--1017}.
\newblock


\bibitem[Yan and Chen(2021)]%
        {yan2021optimal}
\bibfield{author}{\bibinfo{person}{X. Yan} {and} \bibinfo{person}{Y. Chen}.}
  \bibinfo{year}{2021}\natexlab{}.
\newblock \showarticletitle{Optimal Crowdfunding Design}. In
  \bibinfo{booktitle}{\emph{Proceedings of the 20th International Conference on
  Autonomous Agents and Multi-Agent Systems}}. \bibinfo{pages}{1704--1706}.
\newblock


\bibitem[Zubrickas(2014)]%
        {zubrickas2014provision}
\bibfield{author}{\bibinfo{person}{R. Zubrickas}.}
  \bibinfo{year}{2014}\natexlab{}.
\newblock \showarticletitle{The provision point mechanism with refund bonuses}.
\newblock \bibinfo{journal}{\emph{Journal of Public Economics}}
  \bibinfo{volume}{120} (\bibinfo{year}{2014}), \bibinfo{pages}{231--234}.
\newblock


\end{thebibliography}
